\documentclass[]{interact}

\usepackage{epstopdf}% To incorporate .eps illustrations using PDFLaTeX, etc.
\usepackage{subfigure}% Support for small, `sub' figures and tables
\usepackage{float}
\usepackage[sort,numbers]{natbib}% Citation support using natbib.sty
\bibpunct[, ]{(}{)}{;}{a}{}{,}% Citation support using natbib.sty
\renewcommand\bibfont{\fontsize{10}{12}\selectfont}% Bibliography support using natbib.sty

\theoremstyle{plain}% Theorem-like structures provided by amsthm.sty
\newtheorem{theorem}{Theorem}[section]
\newtheorem{lemma}[theorem]{Lemma}
\newtheorem{corollary}[theorem]{Corollary}
\newtheorem{proposition}[theorem]{Proposition}
\newtheorem{assumption}[theorem]{Assumption}
\theoremstyle{definition}

\theoremstyle{remark}

\begin{document}
\pagenumbering{arabic}
%\normalsize

\title{User equilibrium with a policy-based link transmission model for stochastic time-dependent traffic networks} 

\author{
	\name{Hemant Gehlot\thanks{H. Gehlot. Email: hgehlot@purdue.edu} and Satish V. Ukkusuri\thanks{CONTACT S. V.  Ukkusuri. Email: sukkusur@purdue.edu}}
	\affil{Lyles School of Civil Engineering, Purdue University, 550 Stadium Mall Drive, West Lafayette, IN 47907, USA} }
\maketitle %have put it here rather than after \author for double-blind process

\begin{abstract}
Non-recurrent congestion is a major problem in traffic networks that causes unexpected delays during travels. This type of congestion is primarily caused due to stochasticity associated with travel demand and network supply. In such a scenario, it is preferable to use adaptive paths or policies where next link decisions on reaching junctions are continuously adapted based on information gained with time. When there is no incentive for users to change their strategies, an equilibrium is reached. In this paper, we study a traffic assignment problem in stochastic time dependent networks. The problem is modeled as a fixed point problem and existence of the equilibrium solution is discussed. We iteratively solve the problem using the Method of Successive Averages (MSA). A novel network loading model inspired from Link transmission model is developed that accepts policies as inputs for solving the problem. This network loading model is different from the existing network loading models that take predefined paths for input flows. We demonstrate through numerical tests that solving traffic assignment problem with the proposed loading modeling scheme is more efficient as compared to solving the problem using path based network loading models. In addition, the developed model captures traffic realism by modeling link spillovers and shock-wave propagation through LWR kinematic wave theory. Also, we generate a fixed number of policies in each iteration of MSA algorithm where the first policy is optimal and the remaining policies are suboptimal. We propose an algorithm to generate suboptimal policies through which the quality of generated suboptimal policies can be controlled. We show the relevance of the developed suboptimal policy algorithm by proving that the optimal policy is always allocated the largest traffic flow. As a consequence, traffic flow proportion corresponding to the optimal policy can be controlled by suitably varying a parameter in the proposed suboptimal policy generation algorithm. 	
\end{abstract}
\section{Introduction} \label{int}
\noindent Congestion can be caused due to two types of factors: recurrent and non-recurrent. Recurrent congestion usually happens when there is a gap between demand and supply during the peak hours. Whereas, non-recurrent congestion is caused due to unexpected events like incidents, vehicle breakdown, sudden influx of demand etc. Recurrent congestion can be reduced by improving the available supply but the same does not apply to non-recurrent congestion. That is because of the stochasticity associated with non-recurrent congestion. This stochasticity can either come from the demand side like sudden influx of people due to unexpected events like hurricane evacuation, from the supply side due to events like incidents or from both. \\

\noindent One of the popular strategies to reduce non-recurrent congestion is to efficiently use current infrastructure with the advances in information technology like Advanced Traveler Information Systems (ATIS) \citep{yang1998multiple}. These systems can provide information regarding which next road to take, incident locations, predicted travel times for different paths etc. Technologies like loop detectors and cameras gather network wide real time information and assist systems like ATIS in understanding the state of network. However, coverage of loop detectors and cameras on a traffic network might not be sufficient to provide information for all places along the network. For instance, typical spacing of loop detectors is about 500m but traffic
may shuffle significantly between detectors \citep{wilson2008inductance}. In recent times, there has been a steep rise in the availability of big data sources that complement traditional data collection technologies. For example, sources like GPS and cellular data can generate frequent mobility information \citep{song2010limits}. A key big-data insight from this type of data is that repeated observations
over a period of time allows for a richer characterization of mobility-demand processes. These big data sources also allow us to characterize uncertainty on various links better by mining large volumes of historical GPS based data and thus stochastic distributions can be updated in real time \citep{zhan2013urban}. The current advent of the research in Mobility-as-a-Service (MaaS) and big data analytics provides a rich opportunity to consider stochasticity in the supply side decision making.    \\

\noindent If adequate online information is available about the incident and the traveler adapts to it by taking an alternative path, he or she can save travel time compared to the nonadaptive case. These adaptive routes based on the available online information are referred to as policies or strategies. The interaction of stochastic network supply and stochastic demand along with online routing in a time dependent network results into a stochastic dynamic traffic assignment (SDTA) problem. If the users have no incentive in changing their existing strategies, a user equilibrium is said to have reached. This type of assignment problem would be the focus of this paper. \\

\noindent This study is organized in the following manner. Next section reviews the past studies related to traffic assignment and routing in stochastic time dependent networks and outlines the motivation for this study. Section 3 presents the methodology used in this paper. Section 4 presents the conducted numerical results. The final section concludes the study and provides future directions. \\

\section{Background and Motivation} \label{lr}
\noindent  In the past, there has been a growing interest among researchers to solve problems of traffic routing and assignment problems for stochastic time-dependent networks. In dynamic traffic assignment literature, there have been many works using exit functions \citep{merchant1978model,carey1992nonconvexity}, point queue models \citep{vickrey1969congestion,ramadurai2010linear,doan2011existence} and physical queue models \citep{ziliaskopoulos2000linear,ukkusuri2012dynamic,han2011complementarity}. In all these models, users travel along the paths of the underlying network, where path cannot be modified en route (\textit{path} and \textit{route} are used interchangeably in this paper). If congestion levels become very large, users may choose to switch path en route using online information rather than experience larger delays on the path corresponding to their initial choice. \cite{unnikrishnan2009user} develop a formulation for static user equilibrium where users update their route choice in an online manner. \cite{marcotte2004strategic} use the concept of strategies, and network-theoretic representation of hyperpaths for modeling static assignment problems. \cite{ukkusuri2007exploring} develop a methodology for static traffic assignment accounting for user recourse and online information perception. \cite{hamdouch2004strategic} propose a model of dynamic traffic assignment where strategic choices are an integral part of user behavior but consider constant travel delays, which might not be a realistic assumption. \cite{gao2005optimal} proposes a policy based SDTA that iteratively uses a path based network loading model to compute travel time distributions. This method converts policies to paths before utilizing a path based loader. Another method to solve this problem would be to have a traffic flow model that can directly take policies as part of the formulation. Such a model that accepts policies would involve updating current information at each time step of the network loading process. We term this type of network loading as chronological network loading. \\

\noindent In addition, a policy based network loading model should capture realistic traffic features like queue spillover and shockwave propagation. As mentioned before, there are three types of methods to solve traffic assignment problems: exit functions, point queue models and physical queue models. The use of these different traffic models leads to models with varying traffic realism, the physical queuing models being the most realistic since they capture link spillovers and shockwave propagation. Most notable physical queue models include Cell Transmission Model (CTM) \citep{daganzo1994cell} and Link Transmission Model (LTM) \citep{yperman2007link} that are based on the kinematic wave model of traffic flow \citep{lighthill1955kinematic,richards1956shock}. Recently, LTM has been shown to be the most accurate and efficient version of kinematic wave theory \citep{yperman2007link}. But existing LTM algorithms are path based and cannot directly accept policies as input. To the best of our knowledge, there does not exist a sound and an efficient traffic network loading model that can account for link spillovers and shockwaves using kinematic wave theory and does chronological loading by directly accepting policies rather than paths. We address this gap in the literature and propose a policy based network loading model through LTM approach. \\

\noindent In this paper, SDTA problem is solved using the Method of Successive Averages (MSA) approach. This requires generating policies in each iteration of the MSA algorithm based on the updated travel time information from policy based network loading model. There have been many works \citep{hall1986fastest,miller2000least} on adaptive routing choices or policies but most of them do not consider the information gained till the present time to choose next node. \cite{gao2006optimal} study optimal policy problems with stochastic dependency in a time-dependent context, i.e., the decision to take next node is dependent on the triplet: current node, arrival time at current node and current information. \cite{gao2006optimal} also study approximation algorithms for optimal policy routing but their results could be arbitrarily worse in absolute value than those obtained by running the optimal algorithm. Hence, there is a need for developing policy generation algorithms that can generate suboptimal policies whose quality in comparison to the optimal policy can be controlled. \\

\noindent In summary, this paper makes the following contributions:
\begin{itemize}
\item We introduce a chronological network loading model that accepts policies or adaptive paths instead of paths. The proposed traffic flow model is capable of representing spatial queues and shockwaves using LTM approach.
\item Solving SDTA problem with the proposed chronological loading model is shown to be more efficient as compared to the iterative solution approach that uses path based network loaders.  
\item An algorithm to generate suboptimal policies is introduced through which the quality of generated suboptimal policies can be controlled.  
\item It is proved that the optimal policy is always allocated the largest traffic flow as compared to suboptimal policies. Consequently, traffic flow for the optimal policy can be controlled by suitably varying a parameter in the proposed algorithm.
\item The solution existence of SDTA problem is discussed using fixed point theory.
\end{itemize} 
\section{Methodology} \label{method}
\noindent In this section, we first present a modeling system for the policy based SDTA, then discuss the individual components of the model and finally propose a solution heuristic. We shall first introduce the notations that will be used in this paper. Throughout this paper, we use calligraphic script to represent random variables, while the same symbols in traditional script represent instances of the random variables. 
\subsection*{Notation}
\small
\textbf{Indices:} \\
\noindent $r$:  index for realizations \\
$t, t', t''$: indices for time \\
$n$: index for nodes \\
$j$: index for upstream node of a link\\
$k$: index for downstream node of a link\\
$a,b,c$: indices for links \\
$p$: index for paths \\
$\omega$: index for policies \\
$\omega^*$: index for optimal policy \\
$\omega^s$: index for suboptimal policies \\
$o$: index for origin node \\
$d$: index for destination node \\
$l$: index for iterations \\

\noindent\textbf{Parameters:}\\
$R$: number of realizations \\
$\mathcal{D}$: demand distribution\\
$\mathcal{Q}$: flow capacity distribution (or supply distribution) \\
$Q^{r}_{c,t}$: flow capacity of link $c$ for network realization $r$ at time $t$ \\
$L_c$: length of link $c$ \\
$x_a^{L_a}$: downstream end of link $a$ \\
$x_b^0$: upstream end of link $b$ \\
$\delta_{b}^p$: equal to 1 if link $b$ is an element of path $p$, 0 otherwise \\
$\Delta t$: difference between consecutive time steps \\
$T$: simulation time duration\\
$v_{f,c}$: free-flow speed of link $c$ \\
$w_{c}$: backwave speed of link $c$ \\
$k_c$: jam density of link $c$ \\
$p_{a}$: priority constant for incoming link $a$ \\
$K$: maximum number of iterations in the SDTA model \\
$K^{\zeta}$: maximum number of iterations in the iterative network loading model \\ 
$W$: total number of policies in the SDTA model\\
$z_{\omega^s}$: modifying factor for suboptimal policy $\omega^s$ \\
$\kappa$: a parameter in policy choice model \\
$\rho_r$: probability attached with support vector of realization $r$ \\ 
$\xi$: an infinitesimal positive number \\

\noindent\textbf{Sets:} \\
$A_n$: set of incoming links of node $n$ \\
$B_n$: set of outgoing links of node $n$ \\
$P$: set of all paths between the given OD pair \\
$\Omega$: set of policies obtained from the policy generation model \\
$M$: set of all links \\
$U$: set of all nodes \\
$M'$: set of all links in the space-time graph \\
$p_l$: path set at iteration $l$ \\
$\omega_l$: policy set at iteration $l$ \\
$Z(j,t)$: set of all the information available at node $j$ and time $t$ \\
$\Theta(t)$: set of all the events at time $t$ \\
$B(j)$: set of downstream nodes of node $j$ \\

\noindent\textbf{Variables/functions:} \\
$S_{a,t}$: sending flow for link $a$ at time $t$ on downstream end \\
$R_{b,t}$: receiving flow for link $b$ at time $t$ on upstream end \\
$G_{ab,t}$: transition flow from link $a$ to link $b$ at time $t$ \\
$G_{ab,t}^p$: transition flow following path $p$ from link $a$ to link $b$ at time $t$\\
$G_{ab,t}^{\omega}$: transition flow following policy $\omega$ from link $a$ to link $b$ at time $t$\\
$N_t(x)$: aggregate cumulative number of vehicles on location $x$ at time $t$ \\
$N_t^p(x)$: cumulative number of vehicles following path $p$ on location $x$ at time $t$ \\
$N_t^{\omega}(x)$: cumulative number of vehicles following policy $\omega$ on location $x$ at time $t$ \\
$N_{\pi,t}$: cumulative demand at time $t$ \\
$N_{\pi,t}^p$: cumulative demand for path $p$ at time $t$ \\
$N^{-1}(N, x)$: time at which cumulative vehicle number at location $x$ is equal to $N$ \\
$\mathcal{C}$: link travel time distribution \\
$C^r$: link travel times corresponding to realization $r$ \\
$\mathcal{C}_l$: link travel time distribution at iteration $l$ \\
$I_l$: current information at iteration $l$ \\
$\eta_l$: policy splits at iteration $l$ \\
$\mu_l$: path splits at iteration $l$ \\
$C_{a,t}$: travel time value for link $a$ at time $t$ \\
$\delta_{ab,t}^{\omega}$: equal to 1 if link $b$ is chosen by a traveler who follows policy $\omega$ and reaches the node joining links $a$ and $b$ at time $t$, 0 otherwise\\ 
$\theta$: an event \\
$V$: function for translating policies to paths \\ 
$\gamma$: policy generation model\\
$\beta$: policy choice model \\
$\lambda$: policy based network loading model \\
$\alpha$: MSA parameter \\
$e_{\omega}(x)$: expected travel time to destination node $d$ when initial state is $x$ and policy $\omega$ is followed\\
$E_x[y]$: expected value of function $y$ in variable $x$ \\
$Pr(x)$: probability of an event $x$ happening \\
$v_r$: support point of realization $r$ \\
$h_{jk,t} $: realization of $C_{jk,t}$ learned upto current time \\
$Y_{\omega,t}$: utility of policy $\omega$ at time $t$ \\
$C'_{\omega}$: travel time distribution defining policy $\omega$ \\
$f, g$: continuous functions \\
$\chi_{ij}$: equal to 1 if link $(i,j)$ is selected in TDSP problem, 0 otherwise \\
$y_{ij,t}$: equal to 1 if link $(i,j)$ is selected at entering time $t$ in TDSP problem, 0 otherwise 
%\noindent Generate-event-collection() \\
%$D=\{\{v_1,..,v_R\}\}$\\
%For $t=1$ to $T$ \\
%\hspace*{5mm} For each link $(j,k) \in M$  \\
%\hspace*{10mm} For each disjoint set $S \in D$ \\
%\hspace*{15mm} $w$ = number of distinct values among $C^r_{jk,t}$, $\forall r \in S$ \\
%\hspace*{15mm} Divide $S$ into disjoint sets $S'_{1},..,S'_{w}$, such that $C^r_{jk,t}$ is constant over all $r \in S'_i$, $i=1,..,w$ 
%\hspace*{20mm} and $\bigcup_{i} S'_{i} =S$;  \\
%\hspace*{15mm} $D' = D'\{S\} \cup \{S'_{1},..,S'_{w}\}$; \\
%\hspace*{10mm} Next $S$ \\
%\hspace*{10mm} $D = D'$; \\
%\hspace*{5mm} Next $(j,k)$ \\
%\hspace*{5mm} $EV(t) = D$; \\
%Next $t$ \\

\normalsize
\subsection{Policy}
\noindent Before presenting the SDTA model, we introduce the concept of a \textit{policy}. A routing policy is a decision rule that specifies which node to take next at each decision node based on the current time and available online information. In other words, it is a mapping from a set consisting of current node, current time and current information to next node. It distinguishes from path, which is a pre-specified set of successive links between a pair of nodes. Travelers who follow a path make decisions a priori and take a fixed set of links, ignoring online information that provides changing network conditions. \\

\begin{figure}[h]
	\centering
	\includegraphics[width= 0.35\textwidth]{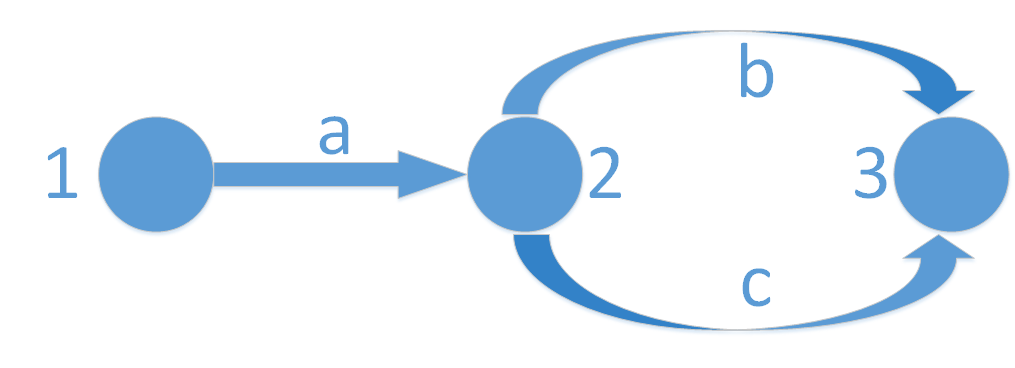} \\
\caption{An example network for explaining the concept of policy}\label{example}
\end{figure}
\begin{table}[h]
	\centering
	\tbl{Time dependent link travel times for the example network}
	{\begin{tabular}{cccc} \hline
		Time& Link&Realization 1 &Realization 2  \tabularnewline \hline
		&$a$&1&2\tabularnewline 
		$t=1$   &$b$&3&1 \tabularnewline
		&$c$&1&5\tabularnewline \hline
		&$a$&4&2\tabularnewline 
		$t=2$   &$b$&4&3 \tabularnewline
		&$c$&5&2\tabularnewline \hline
		&$a$&6&3\tabularnewline 
		$t=3$   &$b$&3&6 \tabularnewline
		&$c$&7&4\tabularnewline \hline
		&$a$&5&8\tabularnewline 
		$t=4$   &$b$&9&2 \tabularnewline
		&$c$&2&1\tabularnewline \hline
	\end{tabular}}
	\label{example_dist}
\end{table}

\noindent We illustrate the concept of policy through an example. Figure \ref{example} and Table \ref{example_dist} present an example network and the corresponding travel time distribution, respectively. The element of the travel time distribution at time $t$ and realization $r$ represents the time a traveler would experience while traveling the link if he or she enters the link at time $t$, provided the network is in realization $r$. We set time interval $\Delta t$ equal to 1 unit. Note that a traveler is aware about the realized travel times of all the links in the network for time steps before time $t$, regardless of his or her current node. This is possible due to various technologies such as observations from GPS traces, probe vehicles in the network, loop detectors and combinations of various big data sources. This type of information access is denoted as the perfect online information (POI) variant \citep{gao2005optimal}. Also, we assume that network stochasticity is characterized by complete link-wise and time-wise statistical dependencies of travel. That is, if a traveler becomes aware of the network realization with full certainty then he or she can deterministically predict the future travel times of all links in the network \citep{gao2006optimal}. At a particular time, the network can experience any one of the two realizations but the travelers are unaware about it. Travelers gain information with the progress of time (with the help of ATIS) and take subsequent decisions based on the collected information. In the example network, if a traveler departing from node 1 at $t=1$ experiences a travel time equal to 2 then ATIS realizes that network is experiencing Realization 2. The traveler reaches node 2 at time $t=3$. So, the traveler would face travel time equal to 6 if she takes link $b$ and would face travel time equal to 4 if she takes link $c$. Hence, on reaching node 2 at time $t=3$, ATIS directs the traveler to travel to link $c$. Instead, if the traveler experiences travel time equal to 1 for link $a$, the network is in Realization 1. Therefore, on reaching node 2 at time $t=2$ traveler is directed to take link $b$ instead of link $c$ as he/she would face a travel time equal to 4 for link $b$ and travel time equal to 5 for link $c$.\\           

\subsection{Model description}
\noindent SDTA model consists of the following three components: policy generation model, policy choice model and policy based network loading model \citep{gao2005optimal}. The inputs to SDTA model are demand distribution $\mathcal{D}=\{D^1,..,D^R\}$ and network supply distribution $\mathcal{Q}=\{Q^1,..,Q^R\}$, where $R$ is the number of network realizations. Note that we present the SDTA model considering a single OD pair that can be easily extended to multiple OD pairs.\\

\noindent Policy generation model $\gamma$ takes link travel time distribution $\mathcal{C}$ as the input and produces a set of policies $\Omega$:
\begin{equation}
\Omega = \gamma(\mathcal{C}) \label{gen}
\end{equation}

\noindent Policy choice model $\beta$ takes policy set $\Omega$ and link travel time distribution $\mathcal{C}$ as the inputs and produces splits (or proportions of traffic flow allocated to different policies) $\eta$:
\begin{equation}
\eta = \beta(\Omega,\mathcal{C}) \label{choice}
\end{equation}

\noindent Policy based network loading model $\lambda$ takes policy splits $\eta$, demand distribution $\mathcal{D}$ and network supply distribution $\mathcal{Q}$ as the inputs and produces link travel time distribution $\mathcal{C}$,
\begin{equation}
\mathcal{C} = \lambda(\eta, \mathcal{D}, \mathcal{Q}) \label{loading}
\end{equation}
Equations \ref{gen}, \ref{choice} and \ref{loading} together form the following fixed problem of the policy-based equilibrium:
\begin{equation}
\mathcal{C} = \lambda(\beta(\gamma(\mathcal{C}),\mathcal{C}), \mathcal{D}, \mathcal{Q}) \label{fixedpoint} 
\end{equation}
The formal definition of policy-based equilibrium is as follows \citep{gao2005optimal}: \textit{A
traffic network is in policy-based stochastic dynamic equilibrium, if each user follows
the routing policy with minimum expected travel time at his/her departure time, and no
user can unilaterally change routing policies to improve his/her expected travel time.} \\

\noindent We solve the fixed point problem using a method of successive averages (MSA) heuristic. That is, at each iteration the travel time distribution is updated by combining the results from the current iteration and previous iterations. The algorithm is as follows: \\ 

\noindent \textit{Step 0 (Initialization)}: \\
0.1 Set $l=0$ \\
0.2 $\mathcal{C}_l$ = free-flow travel times \\
0.3 $l=l+1$ \\ 
\textit{Step 1 (Main step)}: \\
1.1 $\Omega_l = \gamma(\mathcal{C}_l) $ \\
1.2 $\eta_l = \beta(\Omega_l,\mathcal{C}_l) $ \\
1.3 $\mathcal{C'} = \lambda(\eta_l, \mathcal{D}, \mathcal{Q})$ \\
1.4 $\mathcal{C}_l = (1-\alpha)\mathcal{C}_{l-1}+\alpha \mathcal{C'}$, where $\alpha = 1/l$ \\
\textit{Step 2 (Termination check)}:\\
2.1 If $l = K$, then stop. Else, $l=l+1$ and proceed to Step 1. \\

\noindent We study the convergence of this heuristic using empirical results in a later section. In the subsequent sections, we describe the details of individual components of SDTA and discuss solution existence for the fixed point problem. \\

\subsection{Policy generation model}
\noindent Denote $W$ as the predetermined number of policies generated in each iteration of the MSA algorithm. Note that a higher number of policies provides more behavioral strategies to the users but at the cost of increased computation. So, an appropriate number should be chosen based on the application. Policy generation model generates two types of policies: the optimal policy and $(W-1)$ suboptimal policies. Optimal policy generation involves obtaining the optimal policy for the input travel time distribution. Whereas, suboptimal policy generation method is a heuristic inspired from the link penalty method of static shortest path problem \citep{de1993multidimensional}. \\

\subsubsection{Optimal policy generation}
\noindent The optimal routing policy problem in a stochastic time dependent network is to find the policy with minimum expected travel times from all initial states (for all the nodes $j$, all the departure times $t$, and all the information $I$ available at node $j$ and time $t$) to destination node $d$, for a given link travel time distribution. The travel time distribution for which the optimal policy is computed is the distribution that \textit{defines} this policy. We make the following assumptions for the optimal policy problem:  
\begin{enumerate}
\item It is assumed that travelers have access to all link travel time information till the current time (through technologies like ATIS).
\item We assume complete time and spatial dependency between link travel times. 
\item All users follow online information for choosing the next link.
\item Network becomes static and deterministic after the last time step of duration $T$.
\item Different network realizations are independent of each other and the probabilities of their occurrence are known through historical records. 
\end{enumerate}

\noindent Let $e_{\omega}(x)$ denote the expected travel time to destination node $d$ when initial state is $x$ and policy $\omega$ is followed. Denote $B(j)$ as the set of downstream nodes of node $j$, $\mathcal{C}_{jk,t}|I$ as the travel time variable for link $(j,k)$ at time $t$ conditional on current information $I$, and $\mathcal{I'}$ as the current information variable at next node $k$, at time $t+\mathcal{C}_{jk,t}|I$. Denote $Z(j,t)$ as the set of all the information available at node $j$ and at time $t$. Then, for $\forall j \in U-\{d\}$, $\forall t$, $\forall I \in Z(j,t)$, $e_{\omega^*}(x)$ and $\omega^*$ are optimal if and only if they are solutions of the following system of equations \citep{gao2006optimal}:
\begin{equation}
e_{\omega^*}(j,t,I) = \min_{ k \in B(j)} \{E_{\mathcal{C}_{jk,t}}[\mathcal{C}_{jk,t}+E_{\mathcal{I'}}[e_{\omega^*}(k,t+\mathcal{C}_{jk,t},\mathcal{I'})|\mathcal{C}_{jk,t}]|I]\}  \label{optoplicy1} 
\end{equation}
\begin{equation}
\omega^*(j,t,I) = \text{arg} \min_{ k \in B(j)} \{ E_{\mathcal{C}_{jk,t}}[\mathcal{C}_{jk,t}+E_{\mathcal{I'}}[e_{\omega^*}(k,t+\mathcal{C}_{jk,t},\mathcal{I'})|\mathcal{C}_{jk,t}]|I] \}  \label{optpolicy2}
\end{equation}    
with boundary conditions: $e_{\omega^*}(d,t,I) = 0$, $ \omega^*(d,t,I) = d$, $ \forall t$, $\forall I \in Z(d,t)$. The aforementioned system of equations follows the Bellman's principle of optimality \citep{bellman1958routing} as in the static shortest path problem but is extended to take into account current time and information. Therefore, these equations are solved in the similar manner as the shortest path problem.  \\

\noindent We now explain the concept of event, a counterpart of current information, that is more convenient for the implementation of the optimal policy algorithm. Denote $v_r$ as the support point for network realization $r$, $\rho_r$ as the probability associated with $v_r$ such that $\sum_{r=1}^R \rho_r = 1 $. Let $h_{jk,t} $ be the realization of $\mathcal{C}_{jk,t}$ that is learned upto current time. The set of all links in the network is defined by variable $M$. We define event $\theta = \{v_r|C^r_{jk,t'} = h_{jk,t'}, \forall (j,k) \in M, \forall t'<t, \text{for a certain } t\}$. In short, an event at time $t$ contains all the support points that have same travel times across all the links until time $t$. Let $\Theta(t)$ be the set of all the events at time $t$. Also, $\Theta(T)=\{\{v_1\},..,\{v_R\}\}$ because of the Assumption 4 of optimal policy problem.\\

\noindent Denote $e_{\omega^*}(j,t,\theta)$ as the least expected travel time to the destination node $d$ if departure from node $j$ happens at time $t$ with event $\theta$, $\omega^*(j,t,\theta)$ as the next node that should be traveled to realize $e_{\omega^*}(j,t,\theta)$. Then, algorithm \textit{DOT-SPI} for computing optimal policy is as follows \citep{gao2005optimal}: \\ \\
\textit{Step 0 (Event generation step)}: \\
Event Generation;\\
\textit{Step 1 (Initialization step)}: \\
1.1 Compute $e_{\omega^*}(j,T,\theta) \text{ and } \omega^*(j,T,\theta), \forall j \in U, \forall \theta \in \Theta(T)$; \\
1.2 $e_{\omega^*}(j,t,\theta) = \infty, \forall j \in U-\{d\}$, \\
\hspace*{5mm} $e_{\omega^*}(d,t,\theta) = 0$, \\
\hspace*{5mm} $\forall t < T, \forall \theta \in \Theta(t)$; \\
\textit{Step 2 (Main step)}: \\
For $t=T-\Delta t$ down to $\Delta t$ \\
\hspace*{5mm} For each $\theta \in \Theta(t)$ \\
\hspace*{10mm} For each link $(j,k) \in M$ \\
\hspace*{15mm}  $temp = h_{jk,t}+\sum_{\theta' \in \Theta (t+h_{jk,t})} e_{\omega^*}(k,t+h_{jk,t},\theta')  Pr(\theta'/\Theta(t+h_{jk,t}))$; \\
\hspace*{15mm} If $temp < e_{\omega^*}(j,t,\theta)$ \\
\hspace*{20mm} $e_{\omega^*}(j,t,\theta) = temp$; \\
\hspace*{20mm} $\omega^*(j,t,\theta) = k$; \\ \\

\noindent Event Generation in Step 0 involves computing the event matrix by computing events at each time step based on the previously mentioned definition of events \citep{gao2005optimal}. Step 1 is an initialization step. First, the solutions for time $T$ are computed using the static shortest path problem (because of Assumption 4 of optimal policy problem) in Step 1.1 and then Step 1.2 involves initializing the variables for the next step. Step 2 is the main step that applies the Bellman's optimality principle to compute the solutions of Equations \ref{optoplicy1} and \ref{optpolicy2}.         
\subsubsection{Suboptimal policies generation}
\noindent We propose a heuristic based on link penalty method of deterministic shortest path problem to generate suboptimal policies. To generate suboptimal paths from the shortest path in the link penalty method, travel times of the links visited in the shortest path are increased and then shortest path is recomputed for the increased travel times. We use a similar analogy for generating suboptimal policies. \\

\noindent The idea is to increase some elements of the input travel time distribution $\mathcal{C}$ by factors $\{z_{\omega^s}\}, \omega^s \in \{1,..,W-1\}$ and run the optimal policy algorithm \textit{DOT-SPI} on the modified travel time distribution. The factors $\{z_{\omega^s}\}$ are input to the problem whose values are greater than unity and these factors determine how different are suboptimal policies as compared to the optimal policy. Only those elements of the distribution are selected for modification that are part of the optimal policy. We term this algorithm as \textit{LP-policy}. The algorithm is as follows:\\ \\
From $\omega^s=1$ to $W-1$  \\
\textit{Step 1 (Initialization step):} \\
Set $C^r_{jk,t,\omega^s} = C^r_{jk,t}, \hspace{2mm} \forall t, r, (j,k) \in M $\\
\textit{Step 2 (Travel time modification step):} \\
For time $t=T$ \\
\hspace*{5mm} For each event $\theta \in \Theta(t)$ \\
\hspace*{10mm} For each node $j \in U-\{d\}$ \\
\hspace*{15mm} Compute $k = \omega^*(j,t,\theta)$; \\
\hspace*{15mm} Set \hspace{1mm}$ C^r_{jk,t,\omega^s} = z_{\omega^s} C^r_{jk,t}  \hspace{2mm} \forall v_r \in \theta$; \\
\textit{Step 3 (Suboptimal policy assignment step):} \\
suboptimal policy $\omega^s$ = optimal policy obtained using \textit{DOT-SPI} with the modified distribution $C_{\omega^s}$;  \\ 

\noindent In this algorithm, Step 1 is the initialization step where a new travel time distribution that is equal to the original travel time distribution is generated. In Step 2, elements that are part of the optimal policy are increased in the new travel time distribution. Here, $\omega^*(j,t,\theta)$ provides the next node that is guided by the optimal policy $\omega^*$ for state $(j,t,\theta)$. In Step 3, a suboptimal optimal policy is set equal to the optimal policy of the modified travel time distribution.\\

\subsubsection{Complexity of policy generation algorithms} \label{complexitysection}
The complexity of \textit{DOT-SPI} is of $O\left(|M||T|R \ln R + R( SSP)\right)$ where $|M|$ is the number of links, $|T|$ is the number of time periods, $R$ is the number of realizations and $SSP$ is the complexity of solving the static shortest path problem \citep{gao2005optimal}. The complexity of \textit{LP-policy} is of $O(W(R|U|+DS))$, where $W$ is the number of policies, $|U|$ is the number of nodes and $DS$ is the complexity of \textit{DOT-SPI}. Therefore, the overall complexity of policy generation model is $O(W(R|U|+|M||T|R \ln R + R(SSP)))$. Since $O(|M|)>O(|U|)$, the complexity of policy generation model can be simplified to $O(W(|M||T|R \ln R + R( SSP)))$. Therefore, the worst-case complexity of overall policy generation is equal to the number of policies times the complexity of \textit{DOT-SPI}.  
\subsection{Policy choice model} 
\noindent Policy choice model splits traffic flow among the set of policies that are obtained from the policy generation model. We use a logit-based random utility model for computing the travelers' behavior. Denote $\Omega$ as the set of policies obtained from the policy generation model. Let $Pr_t(\omega|\Omega)$ be the probability of choosing policy $\omega$ by a traveler if he or she departs the origin node at time $t$. Then,
\begin{equation}
Pr_t(\omega|\Omega) = \frac{exp( Y_{\omega,t})}{\sum_{\omega \in \Omega} exp( Y_{\omega,t})} \label{Ch1}
\end{equation}   
where $exp(\cdot)$ is the natural exponential function and $Y_{\omega,t}$ is utility of policy $\omega$ at time $t$. We compute $Y_{\omega,t}$ in terms of the expected travel time a traveler takes to reach destination if he or she follows policy $\omega$ and departs at time $t$ from the origin. This expected travel time, $e_{\omega}(o,t)$, is computed as follows:
\begin{equation}
e_{\omega}(o,t)= \sum_{\theta' \in \Theta (t)} e_{\omega}(o, t,\theta')Pr(\theta'/\Theta(t)) \label{Ch2}
\end{equation}  
where $e_{\omega}(o, t, \theta')$ is the expected travel time to destination $d$ if a traveler following policy $\omega$ departs from the origin $o$ at time $t$. After this, utility $Y_{\omega,t}$  is computed as follows: 
\begin{equation}
Y_{\omega,t} = \kappa \hspace{1mm} e_{\omega}(o,t) \label{alpha}
\end{equation}  
where $\kappa$ is a constant. Note that typically, $\kappa$ is negative so that split value increases as the expected travel time reduces.\\

\noindent We apply the large sample approximation to policy choices such that the proportion of travelers taking a given policy is the same as the probability that an individual takes that policy \citep{daganzo1977stochastic}. Let $\eta_{\omega,t}$ denote the policy split for policy $\omega$ at time $t$. Then,  
\begin{equation}
\eta_{\omega,t} = Pr_t(\omega|\Omega) \hspace{10mm} \forall \omega, t \label{splitprob}
\end{equation}  
Note that the split corresponding to the optimal policy is expected to be larger than the splits for suboptimal policies because optimal policy is obtained with the objective of minimizing the expected travel time. Fortunately, this condition holds for the developed \textit{LP-policy} algorithm but might not hold true if the elements of travel time distribution are increased in a different manner in \textit{LP-policy}. We first show that if \textit{LP-policy} is used to compute suboptimal policies, the optimal policy always has the largest split value as compared to the suboptimal policies. 
\begin{proposition} \label{thm1}
\noindent If \textit{LP-policy} is used to compute suboptimal policies then the split corresponding to the optimal policy $\omega^*$ is larger than the splits of all suboptimal policies at each time $t$, i.e.
\begin{equation*}
\eta_{\omega^*, t} > \max\left(\{\eta_{\omega^{s}, t}\}\right) \hspace{10mm} \forall t
\end{equation*}
where $\omega^*$ is the optimal policy, $\omega^s$ is a suboptimal policy and $\{\eta_{\omega^{s}, t}\}$ denotes the set of splits corresponding to the suboptimal policies. 
\end{proposition} 
\begin{proof}
Recall that in DOT-SPI algorithm, the computation of expected travel times is obtained in decreasing order of time. So, we first compare the expected travel times between the optimal policy and suboptimal policies for time $t\geq T$. In algorithm \textit{LP-policy}, travel times of some links are increased at time $T$. Thus, the shortest travel time from any node to the destination would either increase or remain the same because the network becomes deterministic and static for $t\geq T$. That is, 
\begin{equation}
e_{\omega^*}\left(j,t,\theta\right)\leq e_{\omega^s}\left(j,t,\theta\right) \hspace{10mm} \forall j\in U-\{d\}, \theta\in \Theta(t), t\geq T \label{afterT}
\end{equation}  
Note that the optimal policy is obtained using Equations \ref{optoplicy1} and \ref{optpolicy2}. That is, Bellman's optimality principle holds true \citep{bellman1958routing},
\begin{equation}
{e_{\omega^*}(j,t,\theta)}=  \min_{k \in B(j)} \left(h_{jk,t}+\sum_{\theta' \in \Theta (t+h_{jk,t})} e_{\omega^*}(k,t+h_{jk,t},\theta')  Pr(\theta'/\Theta(t+h_{jk,t}))\right) \label{bellman}
\end{equation} 
We now start comparing the expected times between optimal and suboptimal policies starting backwards from time $T-\Delta t$. We compare the terms in the RHS of Equation \ref{bellman} for both the optimal and suboptimal policies. First, $h_{jk,t}, \forall t <T$ is the same for both the original and modified travel time distributions. Second, $e_{\omega^*}(k,t+h_{jk,t},\theta')\leq e_{\omega^s}(k,t+h_{jk,t},\theta')$ from Equation \ref{afterT}. Finally, $Pr(\theta'/\Theta(t+h_{jk,t}))$ remains the same for both the original and modified distributions as probabilities associated with different events remain unchanged\footnote{We only modify travel times for the last time step, so there is no change in event matrix structure for $t<T$. Since network becomes deterministic at the last time step, all support points correspond to different events in the original distribution \citep{gao2006optimal}. Therefore, event structure of the modified distribution does not change at the last time step too.}. Therefore, the following is true from Equations \ref{afterT} and \ref{bellman}:
\begin{equation*}
e_{\omega^*}\left(j,t,\theta\right)\leq e_{\omega^s}\left(j,t,\theta\right) \hspace{10mm} \forall j\in U-\{d\}, \theta\in \Theta(t), t\geq T-1 
\end{equation*}
Similarly, proceeding recursively till the first time step, we obtain the following:
\begin{equation}
e_{\omega^*}\left(j,t,\theta\right)\leq e_{\omega^s}\left(j,t,\theta\right) \hspace{10mm} \forall j\in U-\{d\}, \theta\in \Theta(t), \forall t \label{polessthansub}
\end{equation} 
Now, we compare the expected travel times for all the events at time $t$. For time $t$, $Pr\left(\theta'/\Theta(t)\right)$ would remain the same as mentioned before. Hence, the following is true from Equation \ref{polessthansub}:
\begin{equation*}
e_{\omega^*}(o,t)= \sum_{\theta' \in \Theta (t)} e_{\omega^*}(o, t,\theta')Pr(\theta'/\Theta(t))  \leq \sum_{\theta' \in \Theta (t)} e_{\omega^s}(o, t,\theta')Pr(\theta'/\Theta(t)) \hspace{10mm} \forall t
\end{equation*}
Or,
\begin{equation*}
e_{\omega^*}(o,t) \leq e_{\omega^s}(o,t) \hspace{10mm} \forall t
\end{equation*}
Since $\kappa$ is negative in Equation \ref{alpha}, the following is true:
\begin{equation*}
Y_{\omega^*,t} \geq Y_{\omega^s,t} \hspace{10mm} \forall t
\end{equation*}
The split value of a policy is proportional to its utility by Equation \ref{splitprob}. Therefore,
\begin{equation}
\eta_{\omega^*, t} \geq \eta_{\omega^s, t} \hspace{10mm} \forall t \label{lastlineofproof}
\end{equation}
Since the condition in Equation \ref{lastlineofproof} holds for all the suboptimal policies, the following is true:
\begin{equation*}
\eta_{\omega^*, t} > \max\left(\{\eta_{\omega^{s}, t}\}\right)  \hspace{10mm} \forall t
\end{equation*}
\end{proof}
\noindent Proposition \ref{thm1} can be used to study the effect of factors $\{z_{\omega^s}\}$ on the flow allocated to the optimal policy. The next result presents the variation of optimal policy split with $\{z_{\omega^s}\}$ values.
\begin{corollary} \label{cor:largesplit}
Let there be two vectors $\{z^1\}$ and $\{z^2\}$ such that $z^1_{\omega^s}<z^2_{\omega^s} \hspace{3mm}\forall \omega^s$. Then, the split value of the optimal policy obtained when $\{z^1\}$ is used is larger than the split value obtained when $\{z^2\}$ is used for generating suboptimal policies. That is,
\begin{equation*}
\eta_{\omega^*, t}^{z^1}>\eta_{\omega^*, t}^{z^2}, \hspace{10mm} \forall t
\end{equation*} 
where $\eta_{\omega^*, t}^{z^1}$ and $\eta_{\omega^*, t}^{z^2}$ are optimal policy split values at time $t$ corresponding to vectors $\{z^1\}$ and $\{z^2\}$, respectively.
\end{corollary}
\begin{proof}
For a suboptimal policy $\omega^s$, denote the utility at time $t$ with factors $z^1_{\omega^s}$ and $z^2_{\omega^s}$ as $Y_{\omega^s,t}^{z^1}$ and $Y_{\omega^s,t}^{z^2}$, respectively. Since $z^2_{\omega^s}>z^1_{\omega^s}$, $Y_{\omega^s,t}^{z^2}>Y_{\omega^s,t}^{z^1} \hspace{2mm}\forall t$ by Proposition \ref{thm1} (First, we assume the travel time distribution with the factor $z^1_{\omega^s}$ as the input travel time distribution. Then, we can obtain travel time distribution with factor $z^2_{\omega^s}$ by applying \textit{LP-policy} on the input distribution with factor $z^2_{\omega^s}/z^1_{\omega^s}$).   

Since $Y_{\omega^s,t}^{z^2}>Y_{\omega^s,t}^{z^1} \forall \omega^s, t$ the following is true:
\begin{equation*}
\sum_{\forall \omega^s } exp( Y_{\omega^s,t}^{z^2}) > \sum_{\forall \omega^s } exp( Y_{\omega^s,t}^{z^1}) \hspace{10mm} \forall t
\end{equation*} 
Or,
\begin{equation*}
\frac{exp( Y_{\omega^*,t})}{exp( Y_{\omega^*,t})+\sum_{\forall \omega^s } exp( Y_{\omega^s,t}^{z^2})} < 
\frac{exp( Y_{\omega^*,t})}{exp( Y_{\omega^*,t})+\sum_{\forall \omega^s } exp( Y_{\omega^s,t}^{z^1})} \hspace{10mm} \forall t
\end{equation*}  
Therefore, we can say the following:
\begin{equation*}
\eta_{\omega^*, t}^{z^2}<\eta_{\omega^*, t}^{z^1} \hspace{10mm} \forall t
\end{equation*}
\end{proof}
\noindent Note that Proposition \ref{thm1} may not necessarily hold true if \textit{LP-policy} is modified so that travel times are incremented at times less than $T$. That is because even though it might take longer time for a traveler to traverse link $(i,j)$ using the suboptimal policy (if travel time for link $(i,j)$ is increased in \textit{LP-policy}), his/her travel time from node $j$ to the destination might decrease. That is because when the traveler reaches node $j$ at $t'<T$, optimal expected travel time from node $j$ to destination is dependent on the time $t'$. An example is provided to illustrate this argument in next paragraph. \\

\noindent Consider a suboptimal policy that is obtained by modifying the travel time distribution in Table \ref{example_dist}. For suboptimal policy generation, travel time value of link $a$ at time $t=1$ and Realization 2 is increased from 2 to 3. So, if the network is experiencing Realization 2 then the traveler departing at node 1 at $t=1$ experiences a travel time equal to 3 based on the modified travel time distribution. On reaching node 2 at $t=4$, he or she decides to travel to link $c$. It takes a total of 3+1=4 units to traverse the network as compared to 2+4=6 units to traverse using optimal policy. Since other elements of the distribution are not modified, the expected travel time if the network is in Realization 1 remains the same for both the original and modified distribution. Hence, the expected travel time averaged over both the realizations is larger for the original distribution as compared to the modified distribution. Consequently, larger split value is allocated to the suboptimal policy as compared to the optimal policy. Therefore,  we do not modify travel time distribution at all time steps in \textit{LP-policy}. The next proposition provides sufficient conditions that allows travel time distribution to be increased at arbitrary time steps in suboptimal policy generation but still allocates largest flow to the optimal policy.
\begin{assumption} \label{assum:monoinc}
Input travel time distribution has values that are monotonically increasing function with time. That is, 
\begin{equation*}
C^r_{jk,t'}>C^r_{jk,t''} \hspace{10mm} \forall r \in R, (j,k) \in M, t'>t''
\end{equation*}
\end{assumption}
\begin{proposition}
\noindent If Assumption \ref{assum:monoinc} is satisfied then split value corresponding to the optimal policy is the largest if \textit{LP-policy} involves modifying the elements at arbitrary time steps. That is,
\begin{equation*}
\eta_{\omega^*, t} > \max\left(\{\eta_{\omega^{s}, t}\}\right) \hspace{10mm} \forall t
\end{equation*}     
\end{proposition} 
\begin{proof}
Using the same argument as in the proof of Proposition \ref{thm1}, the following is true:
\begin{equation}
e_{\omega^*}\left(j,t,\theta\right)\leq e_{\omega^s}\left(j,t,\theta\right) \hspace{10mm} \forall j\in U-\{d\}, \theta \in \Theta(t), t\geq T \label{firsteq}
\end{equation}  
Now we compare the expected travel times between the optimal and suboptimal policies for $t<T$. Expected travel times for optimal and suboptimal policies are given as:
\begin{equation*}
{e_{\omega^*}(j,t,\theta)}=  \min_{k \in B(j)} \{h_{jk,t}+\sum_{\theta' \in \Theta (t+h_{jk,t})} e_{\omega^*}(k,t+h_{jk,t},\theta')  Pr(\theta'/\Theta(t+h_{jk,t}))\}
\end{equation*}  
\begin{multline*}
{e_{\omega^s}(j,t,\theta)}=  \min_{k \in B(j)} \{z_{\omega^s,jk,t} h_{jk,t}\\+\sum_{\theta'' \in \Theta (t+z_{\omega^s,jk,t} h_{jk,t})} e_{\omega^s}(k,t+z_{\omega^s,jk,t} h_{jk,t},\theta'')  Pr(\theta''/\Theta(t+z_{\omega^s,jk,t} h_{jk,t}))\} 
\end{multline*}  
where $z_{\omega^s,jk,t}$ is a function of $z_{\omega^s}$ ($z_{\omega^s,jk,t}$ is greater than 1 if travel time of link $(j,k)$ is modified at time $t$, otherwise it is equal to 1). First, we compare the expected travel times of optimal and suboptimal policies at $t=T-\Delta t$. Since $z_{\omega^s,jk,t}\geq1$, we have $z_{\omega^s,jk,t} h_{jk,t} \geq h_{jk,t}$. Note that the event matrices for both the original and modified distributions are equivalent\footnote{The number of events at particular time do not increase in the modified travel time distribution because all support points of a particular event are identically modified in \textit{LP-policy}. However, the number of events can reduce if the events that were originally different become equivalent. In that case, we can still consider them as two different events and the output would be the same.}. Notice that we are now comparing events at different times because of the presence of $z_{\omega^s,jk,t}$ in the last equation. So, consider an event $\theta'$ at time $t+h_{jk,t}$. Now, we choose all events $\theta''$ at time $t+z_{\omega^s,jk,t} h_{jk,t}$ such that these events map all the support points of realizations containing the event $\theta'$ (this is possible as the number of events increase monotonically with time). That is,
\begin{equation}
Pr(\theta'/\Theta(t+h_{jk,t})) = \sum_{\theta''}Pr(\theta''/\Theta(t+z_{\omega^s,jk,t} h_{jk,t})) \label{4}
\end{equation}
Since $z_{\omega^s,jk,t}\geq1$, time $t+z_{\omega^s,jk,t} h_{jk,t}$ is no less than time $t+ h_{jk,t}$. Therefore, the following is true from Equation \ref{firsteq}:
\begin{equation}
e_{\omega^s}(k,t+z_{\omega^s,jk,t} h_{jk,t},\theta'') \geq e_{\omega^*}(k,t+h_{jk,t},\theta') \hspace{10mm} \forall \theta'', t \geq T-\Delta t \label{5}
\end{equation} 
From Equations \ref{4} and \ref{5}, the following is obtained:
\begin{multline*}
\sum_{\theta''}e_{\omega^s}(k,t+z_{\omega^s,jk,t} h_{jk,t},\theta'')Pr(\theta''/\Theta(t+z_{\omega^s,jk,t} h_{jk,t})) \geq \\ e_{\omega^*}(k,t+h_{jk,t},\theta')Pr(\theta'/\Theta(t+h_{jk,t}))  \hspace{10mm} \forall t \geq T-\Delta t
\end{multline*}
If we do the same analysis for all the events $\theta' \in \Theta (t+h_{jk,t})$ then the following holds true:
\begin{equation*}
e_{\omega^*}\left(j,t,\theta\right)\leq e_{\omega^s}\left(j,t,\theta\right) \hspace{10mm} \forall j\in U-\{d\}, \theta\in \Theta(t), t\geq T-\Delta t
\end{equation*}
In fact, the following holds true using the same arguments as above:
\begin{equation*}
e_{\omega^*}\left(j,t,\theta\right)\leq e_{\omega^s}\left(j,t',\theta'\right) \hspace{6mm} \forall j\in U-\{d\}, \theta\in \Theta(t), \theta'\in \Theta(t'), t'\geq t\geq T-\Delta t 
\end{equation*}
Similarly proceeding backwards in time gives the following result:
\begin{equation*}
e_{\omega^*}\left(j,t,\theta\right)\leq e_{\omega^s}\left(j,t,\theta\right) \hspace{10mm} \forall j\in U-\{d\}, \theta\in \Theta(t), \forall t
\end{equation*}  
Proceeding in the same manner as before in Proposition \ref{thm1}, we get the following:
\begin{equation*}
\eta_{\omega^*, t} > \max\left(\{\eta_{\omega^{s}, t}\}\right)  \hspace{10mm} \forall t
\end{equation*} 
\end{proof}  

\subsection{Policy based network loading model} 
\noindent Policy based network loading model consists of finding time-dependent link travel times given policy splits, stochastic demand distribution and stochastic network supply distribution. We use a LTM based approach to develop the policy based network loading model. LTM uses Newell's \citep{newell1993simplified} simplified theory of kinematic waves to propagate traffic on links and evaluates traffic dynamics by means of cumulative vehicle numbers \citep{yperman2007link}. It takes into account traffic properties such as flow and density and thereby captures link
spillovers and shockwave propagation. As in Newell's simplified theory, LTM uses a triangular fundamental diagram. A triangular fundamental diagram of a link $c$ can be defined by three parameters: fixed free-flow speed $v_{f,c}$, backwave speed $w_{c}$ and flow capacity $Q_c$. \\

\noindent As mentioned before, existing LTM algorithms take path splits as inputs rather than policy splits. So we need to have a mechanism for taking policy splits as the input. One approach is to convert policy splits to path splits and use an existing path based network loading model as a black box \citep{gao2005optimal}. We present this approach in the subsequent section using a path based LTM algorithm for the sake of completeness. After this, we propose a novel link based approach of using a chronological network loading model that directly accepts policy splits as input rather than using path splits obtained from converting policy splits.  \\ 

\subsubsection{Path based Link Transmission Model}
We first present the path based LTM algorithm \citep{yperman2007link} that would be used in iterative network loading. We denote this algorithm as \textit{PathLTM}. The algorithm is as follows:  \\ 

\noindent \textbf{\textit{PathLTM:}} \\
\noindent For each time $t$, \\
\hspace{5mm} For each node $n$ at time $t$,

\noindent \textit{Step 1 (Sending and receiving flows computation)}: \\
1.1 For each incoming link $a \in A_n $, compute the sending flow $S_{a,t}$ at the downstream end, $x_a^{L_a}$.    \\
1.2 For each outgoing link $b \in B_n $ determine the receiving flow $R_{b,t}$ at the upstream end $x_b^0$. 

\noindent \textit{Step 2 (Transition flows computation)}: \\
2.1 Compute the aggregate transition flows $G_{ab,t}$ from incoming links $a \in A_n $ to outgoing links $b \in B_n $.\\
2.2 Also, compute disaggregate transition flows $G_{ab,t}^p$ from incoming links $a \in A_n $ to outgoing links $b \in B_n $ for each path $p \in P$. 

\noindent \textit{Step 3 (Cumulative vehicles update)}: \\
3.1 For the downstream boundary of each incoming link $a \in A_n $ and for the upstream boundary of each outgoing link $b \in B_n $ update the aggregate cumulative vehicle numbers:
\begin{equation*}
N_t(x_a^{L_a})= N_{t-\Delta t}(x_a^{L_a})+\sum_{b \in B_n} G_{ab,t} \hspace{10mm} \forall  a \in A_n
\end{equation*}
\begin{equation*}
N_t(x_b^0)= N_{t-\Delta t}(x_b^0)+\sum_{a \in A_n} G_{ab,t} \hspace{10mm} \forall  b \in B_n 
\end{equation*}
3.2 Similarly, update the disaggregate cumulative vehicle numbers:
\begin{equation}
N_t^p(x_a^{L_a})= N_{t-\Delta t}^p(x_a^{L_a})+\sum_{b \in B_n} \delta_{b}^p G_{ab,t}^p \hspace{10mm} \forall  a \in A_n, p \in P \label{disagg1} 
\end{equation}
\begin{equation}
N_t^p(x_b^0)= N_{t-\Delta t}^p(x_b^0)+\sum_{a \in A_n} \delta_{b}^p G_{ab,t}^p \hspace{10mm} \forall  b \in B_n, p \in P \label{disagg2} 
\end{equation}
where $\delta_{b}^p$ is equal to 1 if link $b$ belongs to path $p$, else it is equal to 0. \\

\noindent Note that for simplicity, we consider the same level of time discretization equal to $\Delta t$ for network loading as used in the policy generation model. However, the provided network loading formulation can be easily extended with a finer time resolution for better accuracy. Next, sending flows, receiving flows, transition flows and travel time computation for \textit{PathLTM} are presented: \\

\noindent \textbf{Sending and receiving flows}:\\
Newell's sending flow formulation is as follows \citep{newell1993simplified}:
\begin{equation}
S_{a,t} = min(N_{t'} (x_a^0) - N_{t-\Delta t} (x_a^{L_a}), Q_{a} \Delta t) \label{sending}
\end{equation}
where $t'=t+\Delta t-\frac{L_a}{v_{f,a}}$. \\

\noindent Receiving flow is computed as follows:
\begin{equation*}
R_{b,t} = min(N_{t'} (x_b^{L_b}) +k_b L_b- N_{t-\Delta t} (x_b^0), Q_{b} \Delta t) 
\end{equation*}
where $t'=t+\Delta t-\frac{L_b}{w_{f,b}}$ and $k_b$ is the jam density of link $b$. \\

\noindent \textbf{Aggregate transition flows}: \\ 
The computation of transition flows depends on the type of node. We consider five types of nodes in this study: inhomogeneous, origin, destination, merge and diverge nodes. Figure \ref{nodes} presents the different types of nodes. \\ 

\begin{figure}[h]
	\centering
	\includegraphics[width= 0.9\textwidth]{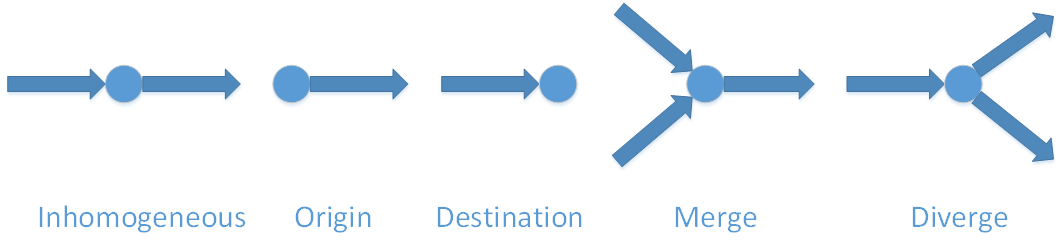}
	\caption{Different types of nodes}\label{nodes}
\end{figure}

\noindent Transition flow of an inhomogeneous node that connects an incoming link $a$ to outgoing link $b$ at time $t$ is given as follows:  
\begin{equation*}
G_{ab,t} = min(S_{a,t},R_{b,t}) 
\end{equation*}

\noindent For origin node, transition flow is:
\begin{equation*}
G_{b,t} = min(N_{\pi,t} - N_{\pi,t-\Delta t},R_{b,t}) 
\end{equation*}
where $N_{\pi,t}$ is cumulative demand at time $t$. \\

\noindent Transition flow for a destination node is simply its sending flow: 
\begin{equation*}
G_{a,t} =  S_{a,t}
\end{equation*}

\noindent Next, we compute transition flows for a merge node. We use the approach by \cite{daganzo1995cell} that assigns priority parameters $p_{a}$ and $p_{a'}$ to incoming links $a$ and $a'$, respectively. Note that $p_{a}+p_{a'}=1$. If the sum of $S_{a,t}$ and $S_{a',t}$ is not greater than the receiving flow of link $b$, then transition flows are given as follows:
\begin{equation*}
G_{ab,t} = S_{a,t} \hspace{10mm} \forall a \in A_n 
\end{equation*}
Otherwise, transition flows are as follows:
\begin{equation*}
G_{ab,t} = median(S_{a,t}, R_{b,t} - S_{a',t} , p_{a} R_{b,t}) \hspace{10mm} \forall a \in A_n 
\end{equation*}

\noindent We use the diverge node transition model given by \cite{daganzo1994cell} in our LTM algorithm:
\begin{equation*}
G_{ab,t} =  min(\frac{R_{b,t}S_{ab,t}}{S_{ab',t}}, S_{ab,t}, R_{b,t})
\end{equation*}
where $b'$ represents the outgoing link other than link $b$ and $S_{ab,t}$ denotes the fraction of sending flow $S_{a,t}$ that wants to go to link $b$ at time $t$. It is computed as follows:
\begin{equation*}
S_{ab,t} =  \sum_{p \in P} \delta_{b}^p (N_{t'}^p (x_a^0) - N_{t}^p (x_a^L)) \hspace{10mm}\forall b \in B_n  
\end{equation*}
where $t'=t+\Delta t-\frac{L_a}{v_{f,a}}$. 

\noindent \textbf{Disaggregate transition flows}: \\ 
In order to compute disaggregate transition flows, we split aggregate transfer flows based on the proportions of sending flow that want to travel to different paths.  \\

\noindent For origin node, disaggregate transition flows are as follows:
\begin{equation}
G_{b,t}^p = G_{b,t}\frac{N_{\pi,t}^p  - N_{\pi,t-\Delta t}^p }{\sum_{p \in P} \left(N_{\pi,t}^p  - N_{\pi,t-\Delta t}^p \right)+\xi} \hspace{10mm} \forall p \in P \label{origindiss}
\end{equation}
where $N_{\pi,t}^p$ is cumulative demand at time $t$ on path $p$ and $\xi$ is an infinitesimal positive number to make sure that the denominator is different from 0. It is computed as follows:
\begin{equation}
N_{\pi,t}^p=\sum_{t'\leq t}D_{t'} \mu_{p,t'} \label{splitflow}
\end{equation}
where $\mu_{p,t'}$ is proportion of flow for path $p$ at time $t'$ and $D_{t'}$ is demand at time $t'$ for the given OD pair. \\

\noindent For destination node, disaggregate transition flows are written as:
\begin{equation*}
G_{a,t}^p = G_{a,t}\frac{N_{t'}^p (x_a^0) - N_{t-\Delta t}^p (x_a^L)}{\sum_{p \in P} \left( N_{t'}^p (x_a^0) - N_{t-\Delta t}^p (x_a^L) \right)+\xi} \hspace{10mm} \forall p \in P
\end{equation*}
where $t'=t+\Delta t-\frac{L_a}{v_{f,a}}$. \\

\noindent Disaggregate transition flows for inhomogeneous and merge nodes are given as follows:
\begin{equation*}
G_{ab,t}^p = G_{ab,t}\frac{N_{t'}^p (x_a^0) - N_{t-\Delta t}^p (x_a^L)}{\sum_{p \in P} \left( N_{t'}^p (x_a^0) - N_{t-\Delta t}^p (x_a^L) \right)+\xi} \hspace{10mm} \forall a \in A_n, p \in P
\end{equation*}

\noindent For diverge nodes, disaggregate transition flows are as follows:
\begin{equation*}
G_{ab,t}^p =  G_{ab,t} \frac{\delta_{b}^p (N_{t'}^p (x_a^0) - N_{t}^p (x_a^L))}{\sum_{p \in P}\left( \delta_{b}^p (N_{t'}^p (x_a^0) - N_{t}^p (x_a^L))\right)+\xi} \hspace{10mm}\forall b \in B_n, p \in P  
\end{equation*}

\noindent \textbf{Travel time computation}: \\
Recall that the travel times computed from the network loading model are used as inputs for the optimal policy algorithm. Hence, we use the following scheme to compute the travel time $C_{a,t}$ for link $a$ at time $t$:
\begin{equation}
C_{a,t}=t-N^{-1}(N_t\left(x_a^L\right ), x_a^0) \label{tt}
\end{equation} 
Here, $N^{-1}(N, x)$ denotes the time at which cumulative vehicle number at location $x$ is equal to $N$. Based on Equation \ref{tt}, travel time of a link at time $t$ is computed as the link time travel experienced by a vehicle that departs the link at time $t$.\\

\noindent It should be noted that the correct computation of link travel times requires First-In-First-Out (FIFO) behavior on links. The presented LTM algorithm ensures approximate link FIFO behavior because sending flow might consist of vehicles having entered the link at different time intervals. However, these effects are small for practical applications and can be minimized by using shorter time intervals \citep{yperman2007link}. \\      

\noindent Note that we might need to compute cumulative vehicle numbers at times that are not multiple of $\Delta t$, for e.g. $t'$ in Equation \ref{sending}. Since \textit{PathLTM} only computes cumulative vehicle numbers at discrete time steps, we use the following interpolation method to compute cumulative vehicle numbers at a time $t'$ that lies in the interval $t$ and $t+\Delta t$: \\   
\begin{equation}
N(t') = N(t) + \left(\frac{t'-t}{\Delta t}\right)\left(N(t+\Delta t)-N(t)\right) \label{tt2}
\end{equation}
\subsubsection{Iterative network loading}
\noindent In this section, we present an iterative network loading model \citep{gao2005optimal} that uses the path based LTM algorithm, \textit{PathLTM}. The algorithm for iterative network loading model is as follows: \\

\noindent For each realization $r$, \\ 
\textit{Step 0 (Initialization step)}: \\
0.1 Set $l=0$ \\
0.2 $C^r_l$ = free-flow travel times \\
0.3 $\{\mu_l, p_l\} = V\left(\eta_l, \omega_l, C^r_l \right)$ \\ 
0.4 $l=l+1$ \\ 
\textit{Step 1 (Main step)}: \\
1.1 $C' = PathLTM(\mu_l, p_l, D^r, Q^r)$ \\
1.2 $C^r_l = (1-\alpha)C^r_{l-1}+\alpha C'$, where $\alpha = 1/l$ \\
1.3 $\{\mu_l, p_l\} = V\left(\eta_l, \omega_l, C^r_l \right)$ \\ 
\textit{Step 2 (Termination check)}:\\
2.1 If $l = K^\zeta$, then $C^r = C^r_l$ and stop. Else, $l=l+1$ and proceed to Step 1. \\

\noindent In the above model, for each realization, $r$, the demand realization $D^r$ and the supply realization $Q^r$ are taken as inputs to the iterative network loading model and travel time distribution $C^r$ is obtained as the output. Step 0 is the initialization step. In Step 0.1, we set iteration index $l$ to zero. Step 0.2 involves setting travel time distribution $C^r_l$ at iteration 0 equal to free flow travel times. Step 0.3 does a translation $V$ from policy set and policy splits to a path set and path splits with free flow travel times as the inputs. We discuss more about this translation in the next paragraph. After this, we increase iteration index by one. Step 1 is the main loop of the algorithm. Step 1.1 involves obtaining travel time distribution from \textit{PathLTM}. Next, we update the travel time distribution using the method of successive averages (MSA) algorithm. In Step 1.3, we convert policy set and policy splits using the travel time distribution obtained from the previous step into path set and path splits, respectively. Finally, we check if the iterations have reached the limit $K^\zeta$ in Step 2. If the limit is reached, we stop and set the travel time distribution from the previous iteration as the output travel time distribution for realization $r$. Otherwise, we proceed to Step 1.  \\

\noindent Now we discuss about translation $V$. The translation algorithm proceeds as follows: \\
For each time $t$,  \\
For each policy $\omega$ at time $t$, 
\begin{enumerate}
\item Set $t'=t$, where $t'$ keeps track of the time as we start traversing from the origin at time $t$ towards the destination.
\item Choose an event $\theta \in \Theta(t)$ such that the difference between current information on travel time distribution, $I_l = \{C^r_{l,ab,t''}| \forall (a,b) \in M, \forall t''<t'\}$ and distribution defining current policy $C'_{\omega}$ is the least. This difference can be defined in terms of the sum of absolute difference between all the elements of $I_l$ and $C'_{\omega}$ till time $t'$.
\item Choose the next node with the obtained event. Update time $t'$, by adding the expected time to travel from current node to the next node by policy $\omega$'s definition to it. 
\item If destination node is reached, stop. Add the policy $\omega$'s split to the obtained path's split value. Else, go to Step 1. \\
\end{enumerate}

\subsubsection{Chronological network loading} \label{chrono}
\noindent As presented in the last section, iterative network loading repeatedly uses \textit{PathLTM} as a black box and employs a translation function for conversion between paths and policies. In this section, we present a network loading scheme that only employs one iteration of network loading. This scheme takes policy set and policy splits as inputs, so it requires development of a novel LTM algorithm. We term this LTM algorithm as \textit{PoLTM}. The algorithm is as follows:\\   

\noindent \textit{\textbf{PoLTM}} \\
For each realization $r$, \\
\textit{Step 0 (Initialization step)}: \\
0.1 Set $t$ =0, \\
0.2 Current information $I_t$ = free flow travel times  \\ 
\noindent \textit{Step 1 (Link-policy incidence matrix computation)}: \\
1.1 For each policy $\omega$, find an event such that the difference\footnote{Difference between the two distributions can be defined in many ways depending on norm used in computing the difference. For instance, we compute the sum of absolute differences of all the elements in the two distributions till the current time.} between travel time distribution defining the current policy $C'_{\omega}$ and current information $I_t$ is the least. \\
1.2 Let $\delta_{ab,t}^{\omega}$ be equal to 1 if link $b$ is chosen by a traveler who follows policy $\omega$ and reaches the node joining links $a$ and $b$ at time $t$, 0 otherwise. Compute $\delta_{ab,t}^{\omega}$, $\forall ab, \omega $, at time $t$ using policies' definition and the events computed in previous step. \\
\noindent \textit{Step 2 (Main step)}: \\
For each node $n$ at time $t$,\\
2.1 For each incoming link $a \in A_n $, compute the sending flow $S_{a,t}$ at the downstream link end ($x_a^L$), and for each outgoing link $b \in B_n $, determine the receiving flow $R_{b,t}$ at the upstream link end ($x_b^0$). \\
2.2 Compute the aggregate transition flows $G_{ab,t}$ from incoming links $a \in A_n $ to outgoing links $b \in B_n $.\\
Also, compute disaggregate transition flows $G_{ab,t}^{\omega}$ from incoming links $a \in A_n $ to outgoing links $b \in B_n $ for each policy $\omega \in \Omega$. \\
2.3 For the downstream link boundary of each incoming link $a \in A_n $ and for the upstream link boundary of each outgoing link $b \in B_n $ update the aggregate cumulative vehicle numbers:  	
\begin{equation*}
N_t(x_a^L)= N_{t-\Delta t}(x_a^L)+\sum_{b \in B_n} G_{ab,t} \hspace{10mm} \forall  a \in A_n
\end{equation*}
\begin{equation*}
N_t(x_b^0)= N_{t-\Delta t}(x_b^0)+\sum_{a \in A_n} G_{ab,t} \hspace{10mm} \forall  b \in B_n 
\end{equation*}

\noindent Similarly, update the disaggregate policy cumulative vehicle numbers:
\begin{equation}
N_t^{\omega}(x_a^L)= N_{t-\Delta t}^{\omega}(x_a^L)+\sum_{b \in B_n} \delta_{ab,t}^{\omega} G_{ab,t}^{\omega} \hspace{10mm} \forall  a \in A_n, \omega \in \Omega \label{polydisagg1}
\end{equation}
\begin{equation}
N_t^{\omega}(x_b^0)= N_{t-\Delta t}^{\omega}(x_b^0)+\sum_{a \in A_n} \delta_{ab,t}^{\omega} G_{ab,t}^{\omega} \hspace{10mm} \forall  b \in B_n, {\omega} \in \Omega \label{polydisagg2} 
\end{equation}
\textit{Step 3 (Travel time update and termination check)}: \\
3.1 Append the obtained link travel times from Step 2 to current information $I_t$. \\
3.2 If $t=T$, then $I_t$ is output travel time distribution for realization $r$ and we stop the algorithm. Else, $t=t+\Delta t$ and proceed to Step 1. \\

\noindent Note that the above algorithm differs from iterative network loading method in the following aspects: 
\begin{enumerate}
\item \textit{PoLTM} takes a policy set and policy splits as the inputs whereas iterative network loading iteratively uses \textit{PathLTM} that takes a path set and path splits as inputs. Also, the update of disaggregate cumulative numbers in \textit{PoLTM} is in terms of policies in comparison to path based disaggregation in iterative network loading. 
\item In \textit{PoLTM}, current information $I$ is updated at each time step as new travel times are obtained. That is why we term \textit{PoLTM} as a chronological network loading model. Here, travelers' decision on choosing the next node is a function of the updated current information. Therefore, $\delta_{ab,t}^{\omega}$ is a function of time unlike $\delta_{b}^{p}$, which is predefined. In iterative loading, current information is updated at the end of each iteration $l$ of the algorithm. The role of current information in iterative network loading comes during the translation of policies to paths using $V$  but there is no role of information in the LTM algorithm as predefined paths are used in \textit{PathLTM}. \\
\end{enumerate}

\noindent It is worth pointing out that in spite of the above stated differences, \textit{PoLTM} follows kinematic wave theory of \cite{newell1993simplified}. The equations for sending flows and receiving flows remain the same as in \textit{PathLTM} and hence we do not present them again. The equations for transition flows and cumulative vehicles update also remain the same except that the disaggregate transition flows and disaggregate cumulative numbers are now updated over polices than on paths. Empirical tests presented in a later section show that chronological network loading is found to be more efficient than iterative network loading on different test networks.

\subsection{Solution existence of fixed point problem}
In this section, solution existence for the fixed point problem of Equation \ref{fixedpoint} is discussed. Solution existence is established using Brouwer's fixed point theorem:
\begin{lemma} \label{brouwer}
\cite{facchinei2007finite} (Theorem 2.1.18): Let $x \subset {\rm I\!R}^n$ be a nonempty convex compact set. Every continuous function $f: x \rightarrow x$ has a fixed point in $x$.
\end{lemma}
\noindent We prove the solution existence with the following assumption:
\begin{assumption} \label{assmp2}
No two network realizations have identical link travel times at a particular time step. That is, for a pair of network realizations $r_1$ and $r_2$,  there exists at least one link $(j,k)$ at each time step $t$ such that $C_{jk,t}^{r_1}\neq C_{jk,t}^{r_2} $.	
\end{assumption}
\noindent Assumption \ref{assmp2} implies that network realizations do not have full overlaps and therefore can be fully determined deterministically\footnote{This is a reasonable assumption as one can expect small differences between any two realizations}. The optimal policy problem under this assumption is similar to the Wait-and-see (WS) or full information problem in literature \citep{gao2006optimal}. Under this assumption, optimal policy problem reduces to multiple time-dependent all-to-one shortest path problems, each corresponding to one of the possible network realizations. With this assumption, we now present the result for solution existence.

\begin{proposition}
If Assumption \ref{assmp2} is satisfied, then the fixed point problem of Equation \ref{fixedpoint} has a solution.
\end{proposition}  
\begin{proof}
We reformulate the fixed point problem of Equation \ref{fixedpoint} for the purpose of simplifying the proof. The fixed point problem is reformulated in terms of policy splits $\eta$ as follows:
\begin{equation*}
\eta=\beta(\gamma(\lambda(\eta)), \lambda(\eta))
\end{equation*}
The input parameters (demand and supply distributions) are ignored in the reformulation since they do not form a fixed point mapping. Note that $\eta$ is a nonempty convex compact set because of the following conditions:
\begin{equation*}
\sum_{\forall \omega}\eta_{\omega,t}=1 \hspace{10mm} \forall t
\end{equation*} 
\begin{equation*}
\eta_{\omega,t}\geq 0 \hspace{10mm} \forall \omega,t
\end{equation*}
We now discuss the continuity of network loading model $\lambda$ in terms of policy splits $\eta$. Consider the function $N=\textit{PoLTM}(\eta)$ where $N$ is the set of cumulative vehicle numbers obtained using \textit{PoLTM} when the set of input policy splits is $\eta$. The function \textit{PoLTM} is implicitly defined from the algorithm and equations of \textit{PoLTM} in Section \ref{chrono}. The continuity of this function can be trivially analyzed as  the equations and steps in \textit{PoLTM} are formed from basic elementary arithmetic operations on continuous functions, which result into continuous functions \citep{rudin1964principles}. Some equations such as those for computing sending flows consist of non-linear functions like $\min$, require careful analysis but they are also continuous because for two continuous functions $f$ and $g$, $\min{(f,g)}=\frac{f+g}{2}-\frac{|f-g|}{2}$ \citep{rudin1964principles}. Therefore, $N=\textit{PoLTM}(\eta)$ is continuous in $\eta$. Also, from Equation \ref{tt} link travel times are continuous functions of cumulative numbers. That is because cumulative numbers are interpolated using a continuous function in Equation \ref{tt2} and inverse of a continuous function on a compact metric space is also continuous \citep{rudin1964principles}. Therefore, the function $\mathcal{C}=\lambda(\eta)$ is a continuous function in $\eta$. \\ 

\noindent Next, $\Omega=\gamma(\mathcal{C})$ is a continuous function if Assumption \ref{assmp2} holds. That is because if Assumption \ref{assmp2} holds, then solving optimal policy problem is equivalent to solving multiple time dependent shortest path (TDSP) problems \citep{hall1993time}. Note that TDSP can be solved by applying static shortest path problem on the expanded space-time graph \citep{pallottino1998shortest}. Since the continuity of shortest path problems is established \citep{chaudhuri2010continuity}, the continuity of TDSP also follows. Interested readers may refer to \cite{chaudhuri2010continuity} for the details related to showing continuity for shortest path algorithms. Alternatively, TDSP  for a single departure time $t^0$ can be formulated as a linear program using a space-time expansion of the physical network:
\begin{equation*}
\min \sum_{(i,j)\in M} \sum_{t \in \{t^0,\ldots,T\}} C_{ij,t}y_{ij,t}
\end{equation*}
\begin{equation*}
\text{subject to: }\sum_{(i,j)\in M}\chi_{ij}-\sum_{(j,i)\in M}\chi_{ji}=
\begin{cases}
1, & i=o \\
-1, &i=d \\ 
0, & otherwise
\end{cases}
\end{equation*}
\begin{equation*}
\sum_{(i,j)\in M} \chi_{ij}\leq 1,\hspace{4mm} \forall i \in U
\end{equation*}
\begin{equation*}
\sum_{(i_t,j_{t'})\in M'}y_{ij,t}-\sum_{(j_{t'},i_t)\in M'}y_{ji,t'}=
\begin{cases}
1, & i=o, t=t^0 \\
-1, &i=d, t=T \\ 
0, & otherwise
\end{cases}
\end{equation*}
\begin{equation*}
\sum_{t \in \{t^0,\ldots,T\}} y_{ij,t}=\chi_{ij}, \hspace{4mm} \forall (i,j)\in M
\end{equation*}
\begin{equation*}
\chi_{ij}\geq 0,\hspace{4mm} \forall (i,j)\in M 
\end{equation*}
\begin{equation*}
y_{ij,t}\geq 0, \hspace{4mm} \forall (i,j)\in M, t \in \{t^0,\ldots,T\}
\end{equation*}

\noindent Here $\chi_{ij}$ is 1 if link $(i,j)$ is selected in the path, 0 otherwise. $M'$ is the set of edges in a space-time graph that is expanded from the physical network and time-varying link travel times. Each physical node $i$ has a node in space time graph corresponding to each time step $t \in \{t^0,\dots,T\}$. Links in space-time graph are governed based on the presence of links in physical network. Travel time values in space-time graph are assigned based on the values $C_{ij,t}$ from travel time distribution. Therefore, $y_{ij,t}$ is equal to 1 if link $(i,j)$ is selected at entering time $t$, 0 otherwise. For more details of the formulation, readers should refer to \cite{yang2014constraint}. Since Karush-Kuhn-Tucker (KKT) conditions represent necessary and sufficient conditions for optimal solutions of a linear program, continuity of TDSP in terms of input travel times can be confirmed by solving the KKT conditions. \\ 

\noindent Also, continuity of policy choice model $\beta$ follows from Equations \ref{Ch1}, \ref{Ch2}, \ref{alpha} and \ref{splitprob} as they involve elementary arithmetic operations on continuous functions. Since a composition of continuous functions ($\lambda$, $\gamma$ and $\beta$) is also continuous \citep{rudin1964principles}, solution existence of the fixed problem follows from Lemma \ref{brouwer}.           
\end{proof}                  
\section{Results} \label{results}
\noindent This section presents the results from conducted numerical tests. First, we present empirical findings on small networks that can be fully interpreted. Later, we present the results on larger networks to show the applicability of our methods to reasonably sized networks.  
\subsection{Small networks}
We conduct tests on two synthetic networks denoted as \textit{TwoLinks} and \textit{Diamond} networks shown in Figures \ref{smallnetworksdia_2link} and \ref{smallnetworksdia_sdmd}, respectively. Tables \ref{2linksdata} and \ref{diamonddata} present the link data for the two networks. In \textit{TwoLinks} network, the OD pair consists of nodes 1 and 3. In \textit{Diamond} network, nodes 1 and 7 form the OD pair. \\

\begin{figure}[h]
	\centering
\includegraphics[width=0.35\textwidth]{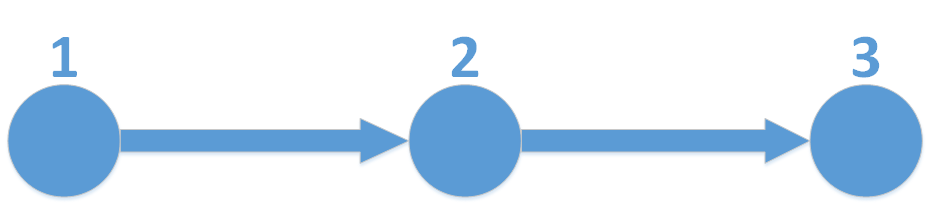}
	\caption{\textit{TwoLinks} network} \label{smallnetworksdia_2link}
\end{figure}
\begin{figure}[h]
	\centering
	\includegraphics[width=0.5\textwidth]{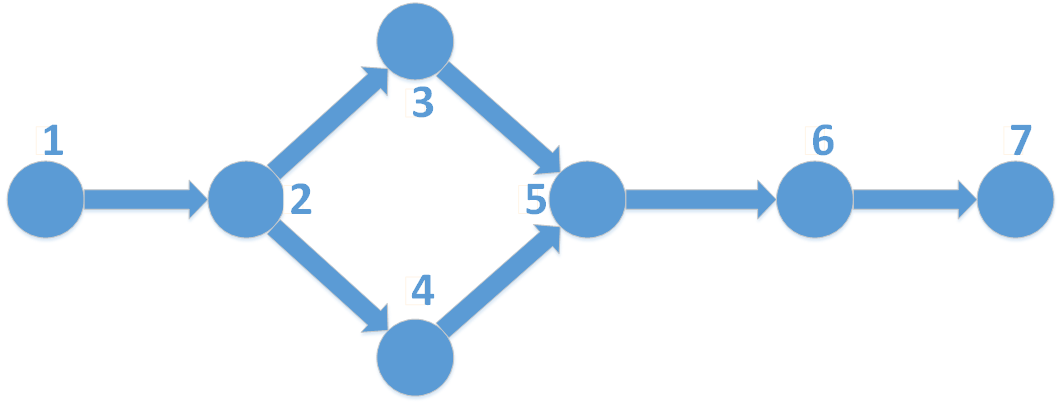}
	\caption{\textit{Diamond} network} \label{smallnetworksdia_sdmd}
\end{figure}

\begin{table}[h]
\centering
\tbl{Link data for \textit{TwoLinks} network}
{	\begin{tabular}{ccc} \hline
		Parameters&Link (1,2) &Link (2,3)   \tabularnewline \hline
		$L_c$ (m)& 860&1220 \tabularnewline \hline
		$v_f,c$ (m/s) & 10&20 \tabularnewline \hline
		$w_c$ (m/s) & 5&10 \tabularnewline \hline
	\end{tabular}}
	\label{2linksdata}
\end{table}
\begin{table}[h]
\centering
\tbl{Link data for \textit{Diamond} network}
{\begin{tabular}{cccccccc} \hline
		Parameters&Link (1,2) &Link (2,3) &Link (2,4)&Link (3,5) & Link (4,5)&Link (5,6)&Link (6,7)  \tabularnewline \hline
		$L_c$ (m)& 860&1220&1360&1220&610&610&610 \tabularnewline \hline
		$v_f,c$ (m/s) & 10&15&20&10&15&20&20 \tabularnewline \hline
		$w_c$ (m/s) & 5&10 &10&5&7.5&10&10\tabularnewline \hline
	\end{tabular}}
	\label{diamonddata}
\end{table}

\noindent We consider three realizations for both the networks with each realization having an equal chance of occurring. For both the networks, we consider supply stochasticity in the link connecting nodes 2 and 3. For this link, we consider three traffic states corresponding to the three network realizations: (i) normal, (ii) congested and (iii) highly congested. In the congested and highly congested states, capacity of the link reduces by 50\% and 75\%, respectively. The remaining links operate in normal conditions. However, there are small variations in capacities of all the links with time regardless of the variations with network realizations. Also, we generate demand distribution where each element of the distribution $D_t^r$ is a uniformly randomly generated valued from 4000 to 4100 vehicles per hour in the first half of simulation duration and from 4000 to 5000 vehicles per hour in the second half. We set the time resolution $(\Delta t)$ to be 1 second and simulate for 600 time steps. We implement the code in Java and run on a Intel Core i7 processor with 3.4 GHz CPU speed and 16 GB RAM.       \\

\noindent We divide our results into five categories: convergence study, computational times comparison of different loaders, variation of average travel time with traffic states, sensitivity analysis of average travel time with the level of variation and sensitivity analysis of policy splits towards perturbation in travel time distribution. \\
\subsubsection{Convergence study}
\noindent In this section, we study the convergence of SDTA algorithm using the two aforementioned types of network loading models. These results build the foundation for the later results. Note that we set the maximum number of iterations for iterative loader to $5$ as we find that after 5 iterations the relative differences in travel times reduce to less than $5$\%. We now define the convergence of SDTA algorithm. Recall that in SDTA model, a MSA algorithm is used for updating travel time distribution in each iteration. We check the relative differences of time-dependent policy splits from two successive iterations. We say that convergence is reached if the relative differences become small. This would imply that users stick to their choice of policy and a policy based equilibrium is reached. \\

\begin{figure}[h]
	\centering
	\includegraphics[width=\textwidth]{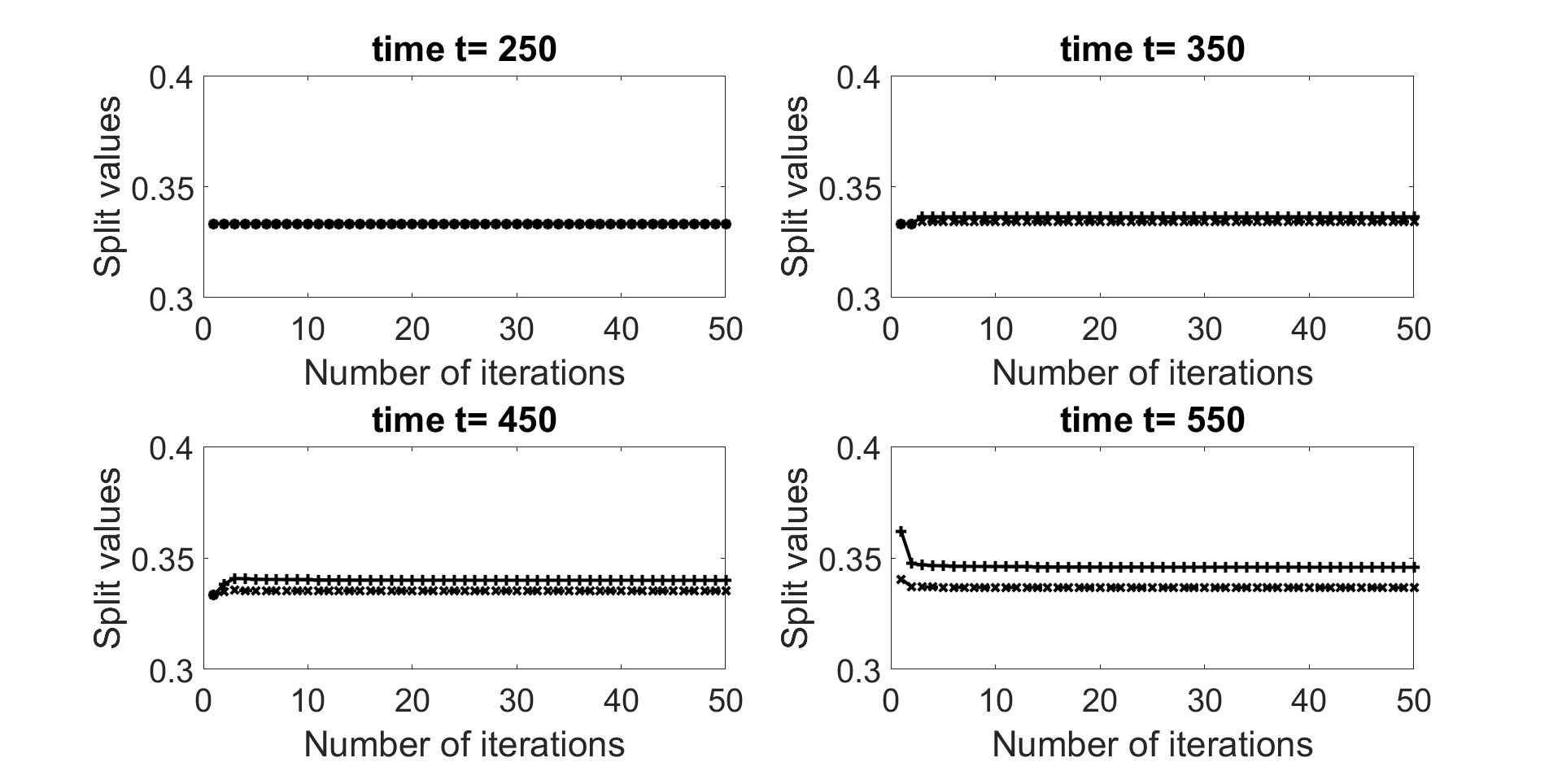}
	\caption{Convergence results for \textit{TwoLinks} network (solid line is split of first policy (the optimal policy) in iterative loading, cross signs represent the split of second policy in iterative loading, plus signs represent splits of first policy in chronological loading and dashed line is split of second policy in chronological loading)}\label{conv_2link}
\end{figure}

\begin{figure}[h]
	\centering
	\includegraphics[width=\textwidth]{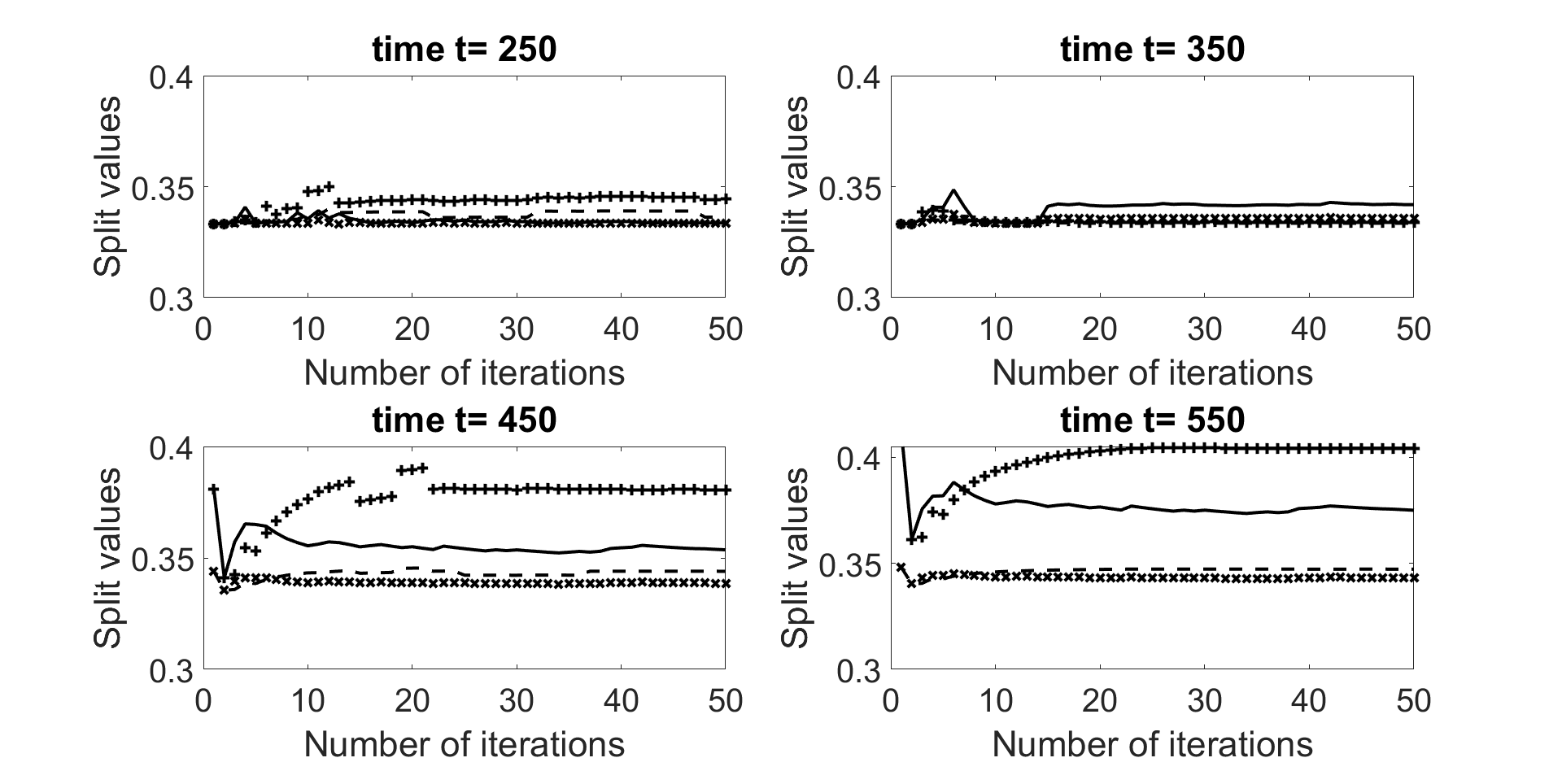}
	\caption{Convergence results for \textit{Diamond} network (solid line is split of first policy in iterative loading, cross signs is split of second policy in iterative loading, plus signs is split of first policy in chronological loading and dashed line is split of second policy in chronological loading)}\label{conv_dia}
\end{figure}

\noindent Figures \ref{conv_2link} and \ref{conv_dia} present convergence results for \textit{TwoLinks} and \textit{Diamond} networks, respectively. We present the results for four time periods (250, 350, 450 and 550) and keep the number of policies per iteration of SDTA algorithm to three. So, we only present the policy splits for the first two policies as the sum of all the policy splits is equal to one. We observe that policy splits converge by 50 iterations for both the loaders as the absolute difference between policy splits of consecutive iterations becomes less than 0.001. Notice that the first policy always has larger split as compared to the second policy for both the loaders. That is because first policy is the optimal policy and therefore has the largest split from Proposition \ref{thm1}. It can also be observed that difference in split values of first and second policies is smaller for time $t=250$ as compared to later times. That is because perturbation of travel time distribution in last time step may not influence the travels starting with early departure time.        \\ 

\noindent Next, we observe that splits of a policy from the two types of the loaders are exactly the same for \textit{TwoLinks} network. This is because there are no diverge nodes in \textit{TwoLinks} network and hence there is no decision making while traversing from origin to destination. Since there is no decision making involved, paths and policies become equivalent. So, it does not matter if the solution is iteratively computed using a path based loader or computed using a chronological policy based loader. For \textit{Diamond} network, though the splits of a policy from both the loaders are not identical but differences in the split values after 50 iterations become small. The maximum observed difference is 0.02 in the final splits from the two loaders. Therefore, the final solutions obtained from both the loaders are similar.     \\    

\subsubsection{Computational performance of different loaders}
\noindent In this section, computational performance of the two aforementioned loaders is discussed. Table \ref{times} presents the computational times while keeping the number of iterations to be the same in both the loaders. First column presents the computational times for iterative loading, second column presents computational times for chronological loading and last column presents the values of second column multiplied by the number of maximum iterations of iterative loading (equal to 5 in our case). It is clear that chronological loading is more efficient than iterative loading. An interesting observation that can be seen is that iterative loading times are higher than the values obtained by multiplying chronological times with the maximum number of iterations of \textit{PathLTM} in iterative loading. This difference arises because of computationally expensive translation function $V$ in iterative loading that involves iterating from the origin to the destination till a path is obtained. This hypothesis is supported by the fact that the difference between the first and third columns is very low for \textit{TwoLinks} network as compared to the same difference for \textit{Diamond} network. That is because while traversing from the origin to the destination in \textit{TwoLinks} network we only to need to cross one node as compared to crossing four nodes in \textit{Diamond} network. Thus, gap between the computation performance of two loaders increases with network size. Hence, we believe that chronological loading is more efficient than iterative loading. Since, the solutions from both the loaders are found to be similar from the previous section we conclude that it is better to use chronological loading over iterative loading. Therefore, from this point onwards we only provide results for chronological loading.\\     
           
\begin{table}[h]
\centering
\tbl{Computational times (in seconds) for different types of loaders and networks}
{\begin{tabular}{cccc} \hline
		Network&Iterative loading  &Chronological loading  & Chronological loading*5   \tabularnewline \hline
		\textit{TwoLinks} network& 62.0 & 6.8&34\tabularnewline \hline
		\textit{Diamond}  network& 1949.5&52.4&262\tabularnewline \hline
	\end{tabular}}
	\label{times}
\end{table} 

\subsubsection{Variation of average travel time with traffic states}
In this section, we consider the effect of stochasticity in link capacities on average expected travel time. Average expected travel time is the expected time it takes to travel from the origin to the destination if optimal policy is followed but averaged over the duration of simulation. Figure \ref{capacity} shows the plot of average expected travel time with different traffic states for \textit{Diamond} network. As mentioned before, we only consider fluctuations in capacity over different realizations for link (2,3). In the first case, we keep the mean capacity of the link to be same for all realizations and there are only small perturbations in the capacity at different time intervals. Later, we gradually reduce capacity of the link for congested and highly congested network realizations. We find that average expected travel time increases as different network states become more congested. This is in accordance with the expectation that the average expected travel time should reduce even if some of the possible network states become more congested.        

\begin{figure}[h]
	\centering
	\includegraphics[width= 0.8
	\textwidth]{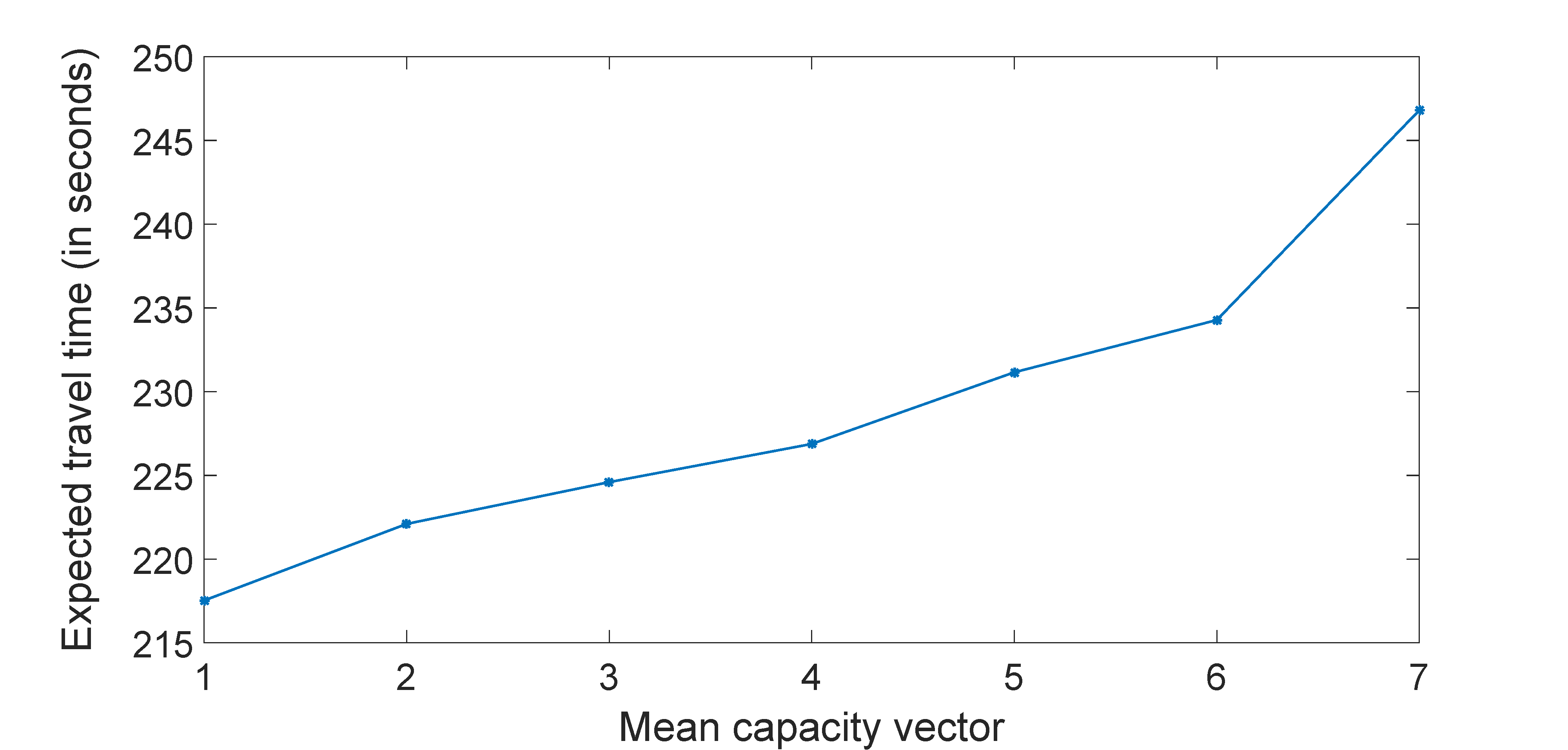}		
	\caption{Variation of average travel time with traffic states (A point on x-axis denotes a vector [a,b,c] where a, b and c are mean capacities (vehicles/second) of link (2,3) in network realizations 1, 2 and 3, respectively. Points 1-7 represent mean capacity vectors [1,1,1], [1,1,0.5], [1,0.9,0.45], [1,0.8,0.4], [1,0.7,0.35], [1,0.6,0.3] and [1,0.5,0.25], respectively. Y axis represents expected time (in seconds) to travel from origin to destination if optimal policy is used but averaged over the whole simulation duration)}\label{capacity}
\end{figure}
\subsubsection{Sensitivity analysis of variation in average travel time with the level of variation}
We now test the sensitivity of the model towards varying levels of variation. One way to quantify the level of variation is through coefficient of variation, which is the ratio of standard deviation to the mean of a probability distribution. Recall that we have variation in the capacities of links at different times of network realizations. Similarly, we also have stochasticity in the demand side. We compute the effect of varying the coefficient of variation (i.e., the level of stochasticity) on the average expected travel time (as defined in the last section). Figure \ref{levelofvar} presents the plot of standard deviation of the average expected travel time with the coefficient of variation for two times during the simulation for \textit{Diamond} network. We compute the variation in the average expected travel time by simulating multiple scenarios from demand and supply distributions and then compute the standard deviation in the end, a technique commonly known as Monte Carlo Simulation \citep{gehlot2016formulation}. It can be seen that as the coefficient of variation increases, the standard deviation of the expected travel time also increases. This is also expected as more stochasticity would produce more uncertainty in the travel times. \\

\begin{figure}[h]
	\centering
	\includegraphics[width= 0.8
	\textwidth]{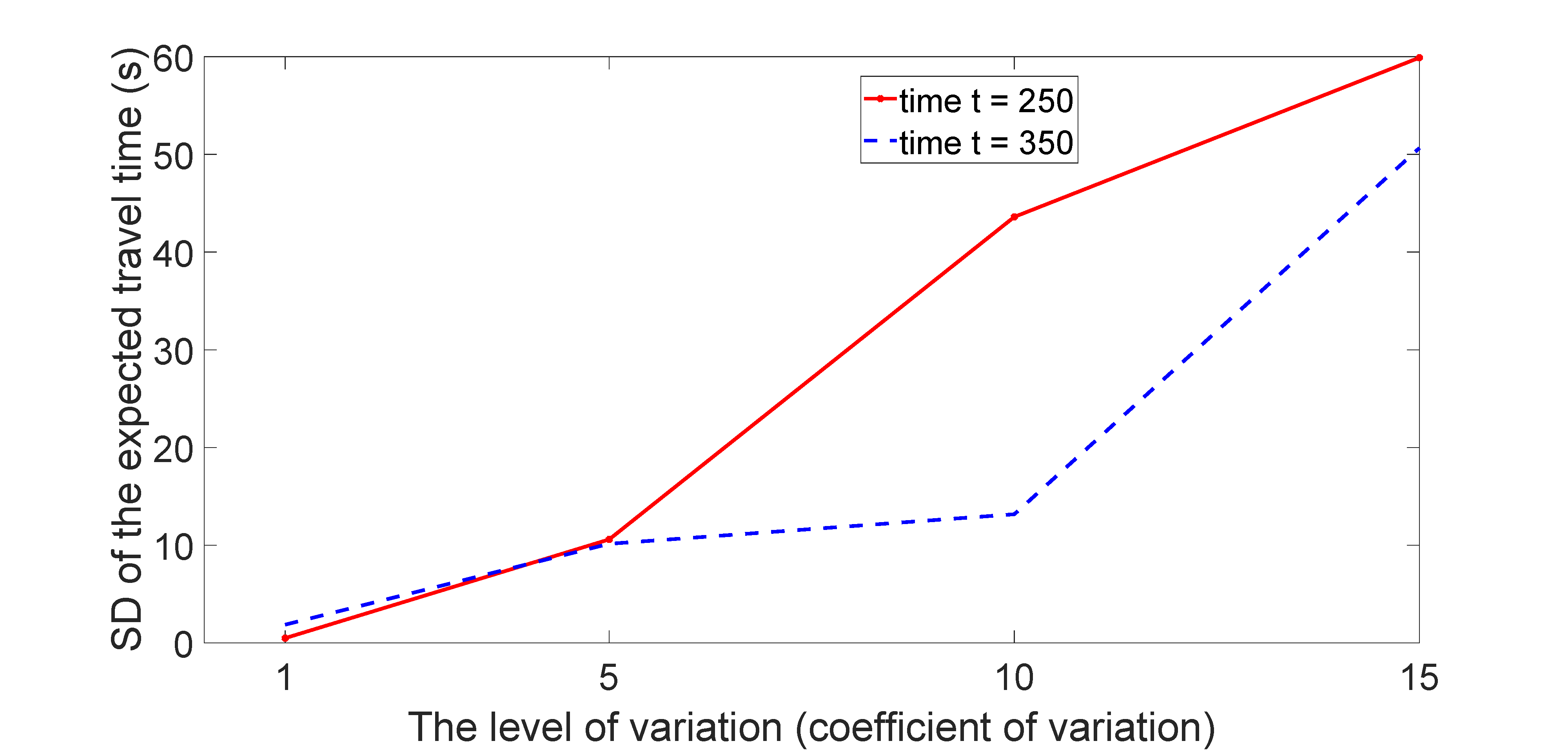}		
	\caption{Variation of standard deviation in the average travel time with different levels of variation (x-axis represents the coefficient of variation that we use in the demand and supply distributions and y-axis represents the observed standard deviation in the expected travel times at two times (equal to 250 and 350) in the simulation)}\label{levelofvar}
\end{figure}
\subsubsection{Sensitivity analysis of policy splits towards perturbation in travel time distribution during suboptimal policies generation}
\noindent We now discuss how policy splits vary with the level of perturbation that is introduced in the travel time distribution while generating suboptimal policies. Figure \ref{policyfac} provides the variation of split of 1st policy (considering two policies in SDTA algorithm) for different levels of perturbation in original travel time distribution in case of \textit{Diamond} network. The level of perturbation is measured in terms of the factors $\{z_{\omega^s}\}$ in \textit{LP-policy} algorithm. Since we only consider one suboptimal policy in this experiment, we are interested in $z_1$. When $z_1$ is 1, there is no change in travel time distribution and split values for both the policies are the same and equal to 0.5. But as $z_1$ increases, split value for the $1st$ policy increases as gap between the expected travel times of optimal and suboptimal policies increases (in accordance with Corollary \ref{cor:largesplit}). Therefore, factors $\{z_{\omega^s}\}$ can be used to control the quality of suboptimal policies as compared with optimal policy. \\

\begin{figure}[h]
	\centering
	\includegraphics[width= 0.8\textwidth]{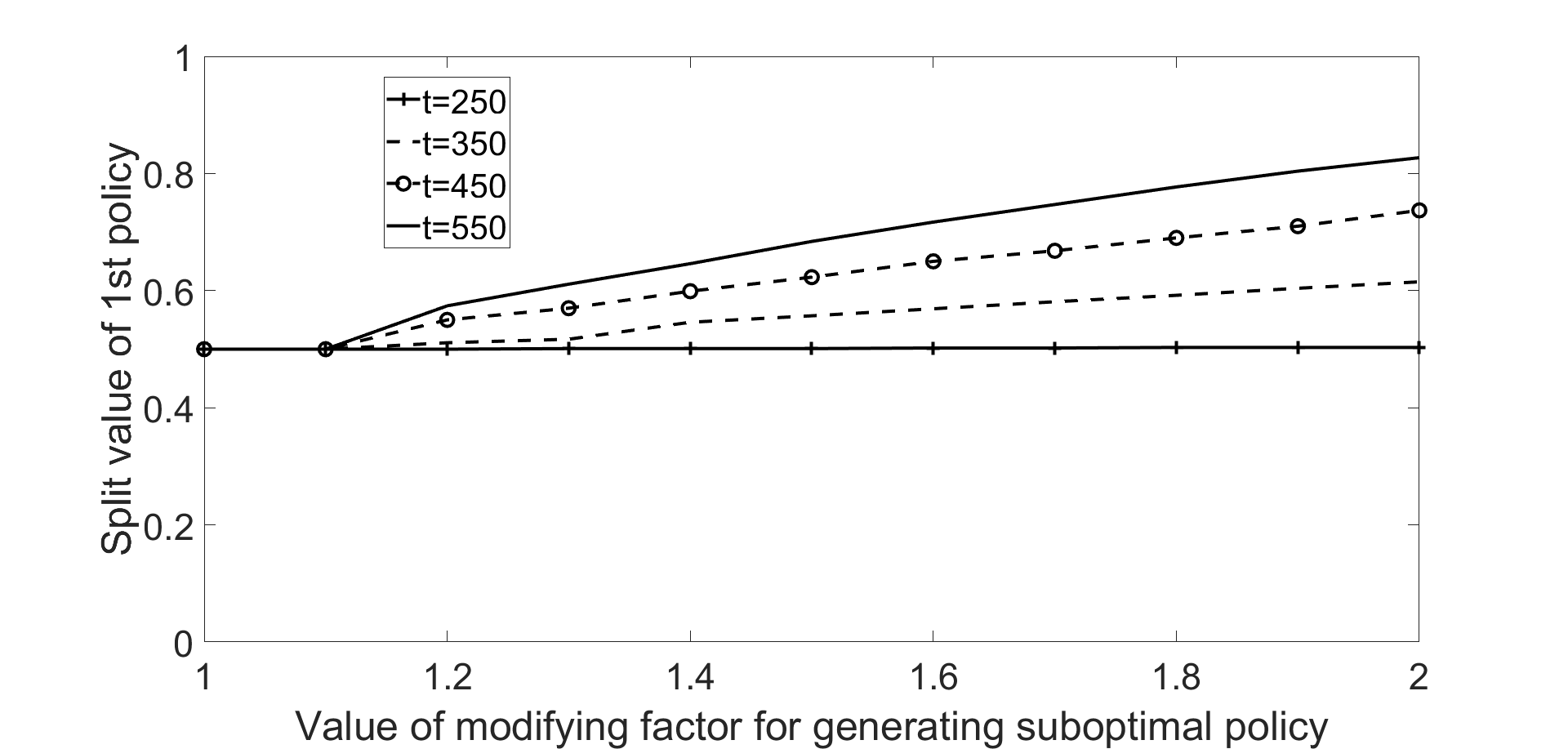}	
	\caption{Variation of split of first policy with the level of perturbation in travel time distribution for \textit{Diamond} network (four lines represent the split values at different times of the simulation)}\label{policyfac}
\end{figure}
\subsection{Modified Sioux Falls network variants}
In this section, we present the computation results on larger networks to check the scalability of the presented SDTA algorithm. Figure \ref{largenet_SF} presents the network which is a variant of Sioux Falls network \citep{ukkusuri2012dynamic}. We modified the original network to restrict the node types to the five types discussed in the Methodology section. We term this network as \textit{SF} network. It has 26 nodes and 36 links with node 1 as the origin and node 26 as the destination. Figure \ref{largenet_2SF} presents the network formed by combining two \textit{SF} networks and is termed as \textit{TwoSF} network.  It has 52 nodes and 73 links with node 1 as the origin and node 52 as the destination. The results on these two networks are compared to understand the computational performance with the size of networks. We first present the convergence results when chronological loading is used to solve the SDTA problem for these two networks. Figure \ref{sf,2sf_conv} presents the convergence of policy splits for two time periods when there are 3 policies in SDTA model. Convergence is reached as absolute differences in split values for consecutive iterations reduces below 0.001. \\   

\begin{figure}[h]
	\centering
	\includegraphics[width=0.45\textwidth]{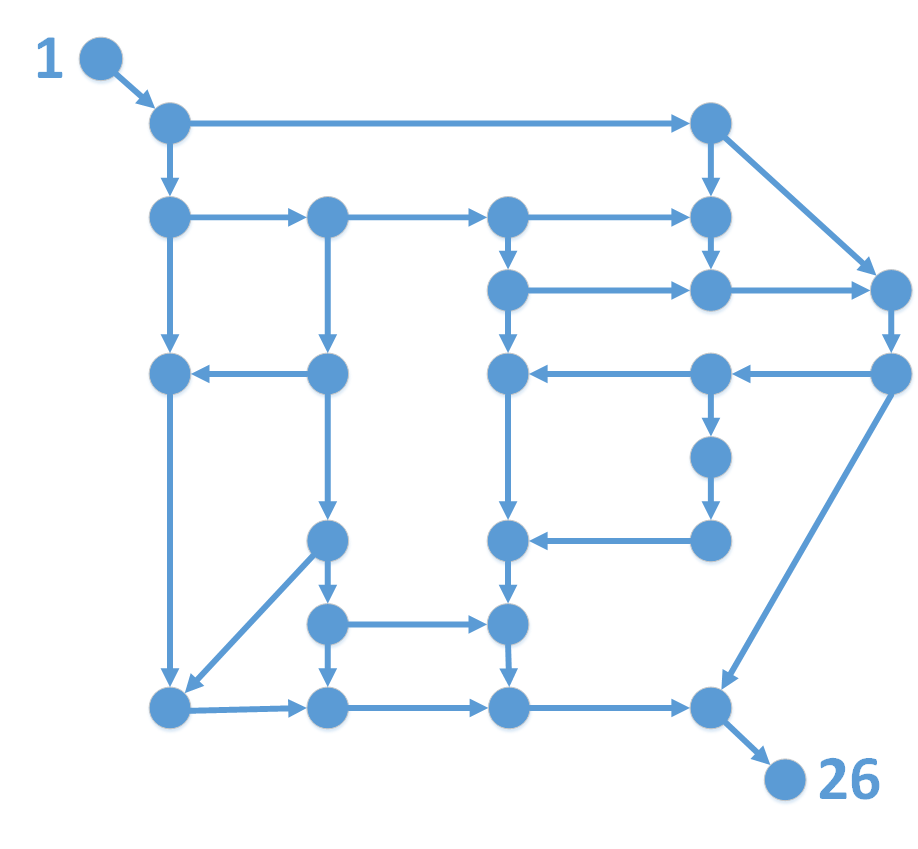}	
	\caption{ \textit{SF} test network}\label{largenet_SF}
\end{figure}
\begin{figure}[h]
	\centering
	\includegraphics[width=\textwidth]{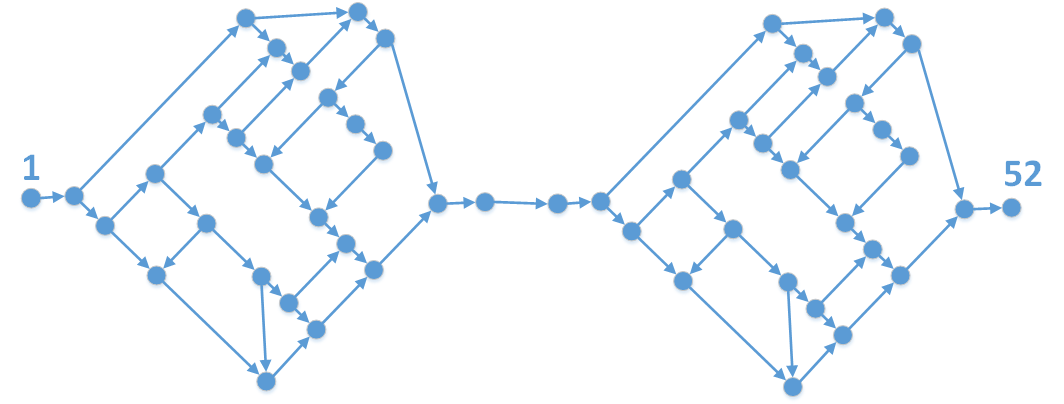}	
	\caption{\textit{TwoSF} test network}\label{largenet_2SF}
\end{figure}

\begin{figure}[h]
	\centering
\includegraphics[width= \textwidth]{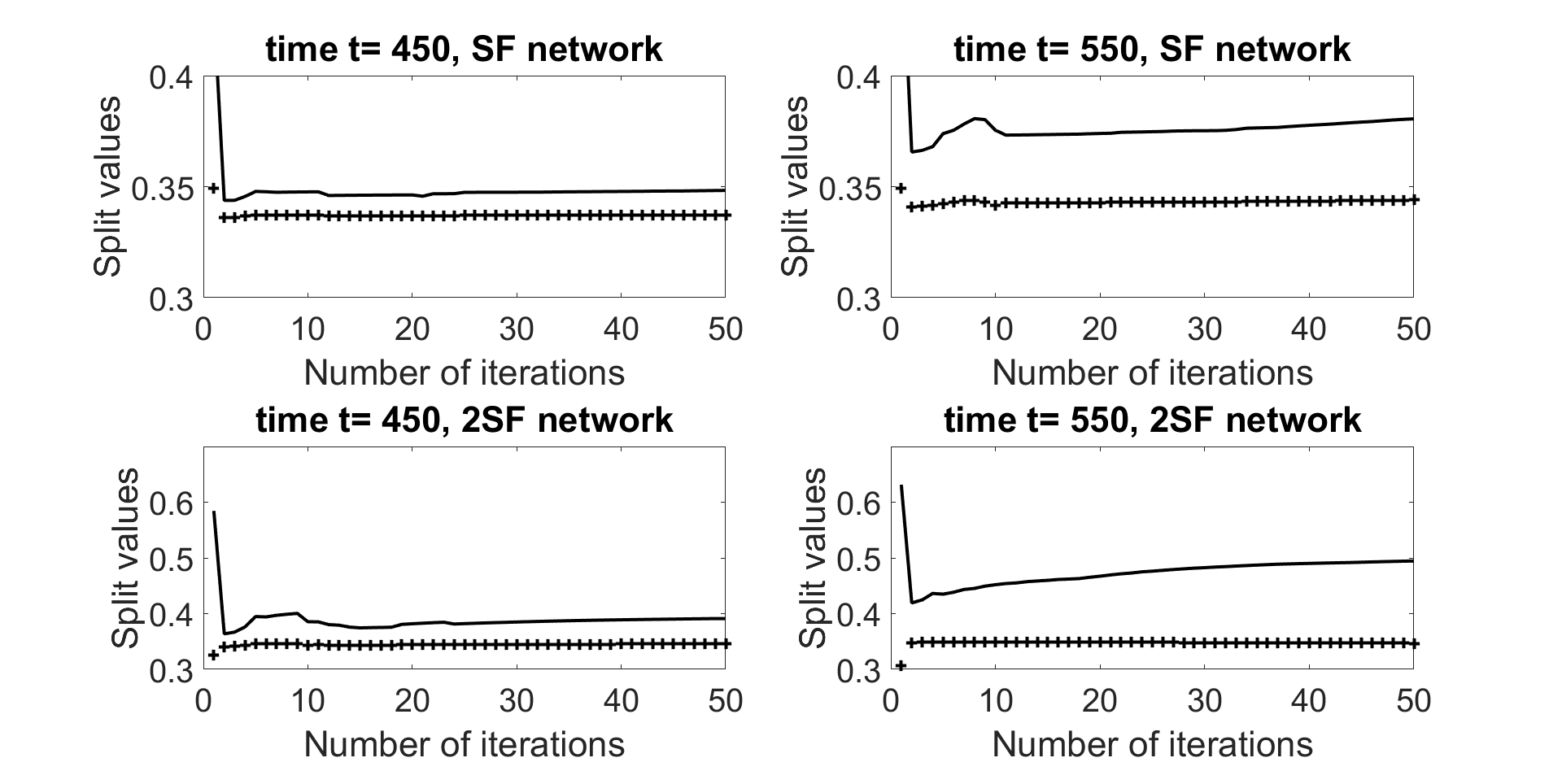}	
	\caption{Convergence results for \textit{SF} and \textit{TwoSF} networks (solid line represents policy split values of first policy and plus signs represent the policy split values of second policy)}\label{sf,2sf_conv}
\end{figure}

\noindent Next, we present the variation of computation time with the number of policies and then discuss the variation of computation performance with the number of network realizations.
\subsubsection{Variation of computation time with policies}
Figure \ref{var_policies} presents the plots of computation time with the number of policies considered in SDTA algorithm. It can be seen that computation time grows in an approximately linear fashion with the number of policies for both the networks. Recall that in Section \ref{complexitysection}, policy generation algorithm is shown to be linear in the number of policies. Also, policy choice model's computation varies linearly with the number of policies. Finally, network loading model is at most linear with the number of policies because algorithm steps like computation of disaggregate transition flows and disaggregate cumulative vehicles updates are linear with the number of policies. Therefore, the observed variation is in accordance with the complexity of SDTA algorithm. \\

\begin{figure}[h]
	\centering
\includegraphics[width= \textwidth]{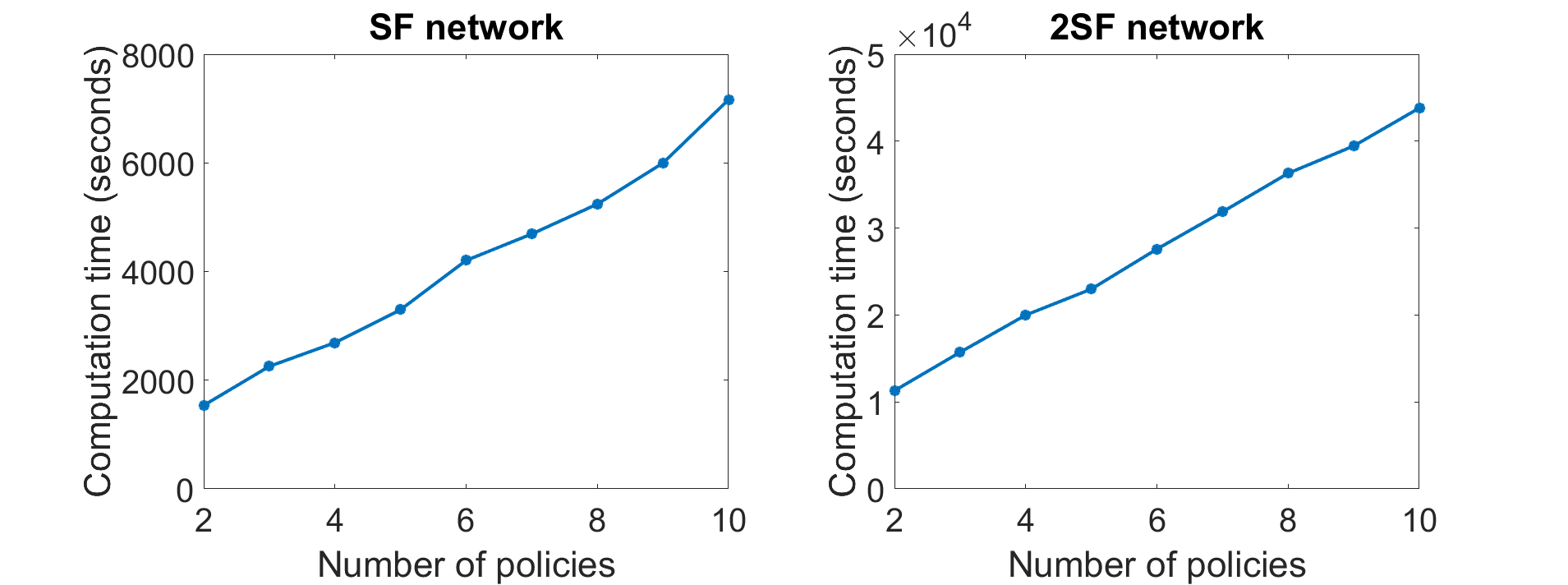}	
	\caption{Variation of computation time with the number of policies for \textit{SF} and \textit{TwoSF} networks}\label{var_policies}
\end{figure}
\subsubsection{Variation of computation time with network realizations}
Figure \ref{var_real} presents the plots showing the variation of computation time with the number of network realizations. It can be seen that the variation is steeper than the variation observed with the number of policies and is non-linear in nature. Note that though network loading algorithm is linear with the number of network realizations the same is not true for policy generation algorithm. As discussed in Section \ref{complexitysection}, policy generation algorithm is $O(R \ln R)$ in the number of network realizations. Therefore, we see the observed pattern for computation time. However, the variation of computation time with the number of realizations is still polynomial\footnote{If link travel times are highly independent then the number of realizations can be exponential in the number of links in the network \citep{gao2005optimal}. However, we assume full link dependency as mentioned before. }. \\

\begin{figure}[h]
	\centering
\includegraphics[width= \textwidth]{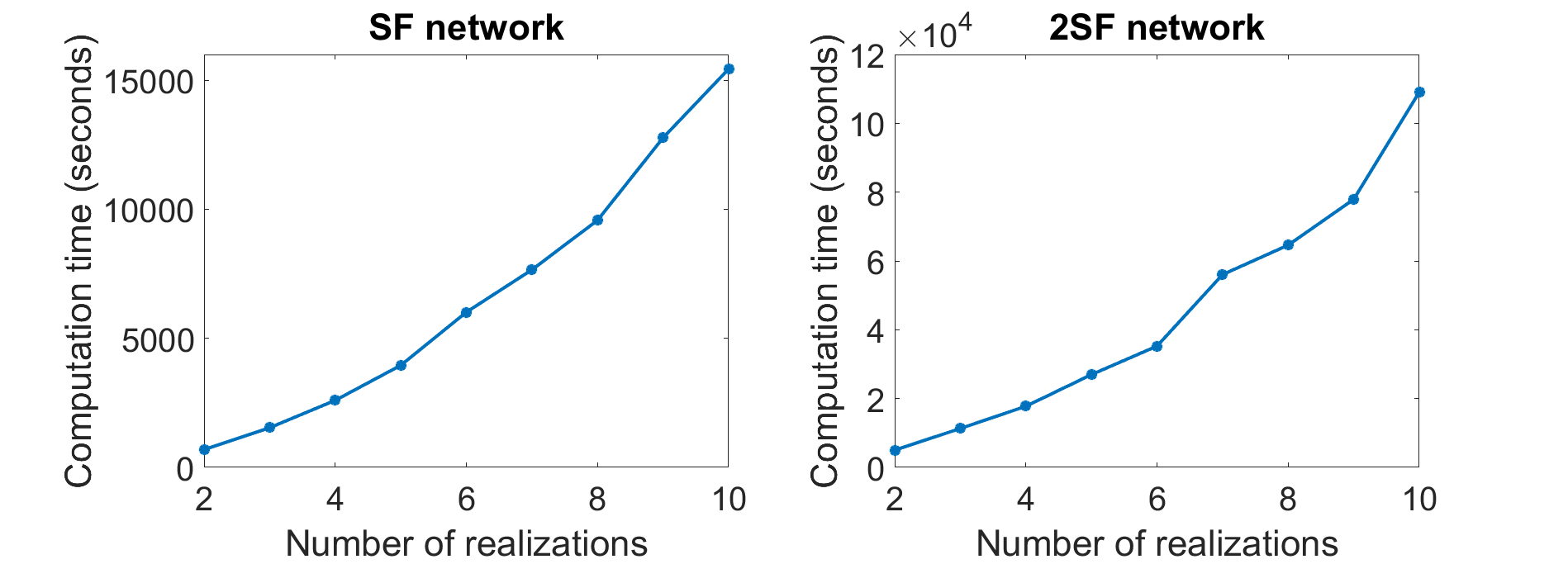}	
	\caption{Variation of computation time with the number of realizations for \textit{SF} and \textit{TwoSF} networks}\label{var_real}
\end{figure}

\noindent We observe that computation performance of SDTA algorithm is polynomial with both the number of policies and network realizations. However, we observe that computation scales significantly with the size of the network. This is because of the increased computation in policy generation and network loading due to increased number of nodes and links. Therefore, further improvement of computational performance through parallelization techniques needs to be explored in future works.        

\section{Conclusions and Summary} \label{conclusions}
\noindent This paper studies an approach to solve the traffic assignment problem for stochastic time dependent traffic networks. The solution algorithm of SDTA problem involves three components: policy generation model, policy choice model and policy based network loading model. In the past, policy based network loading model has been solved by iteratively using a path based network loading model. This requires converting paths to policies before feeding policies into a path based loader. In this study, we propose a policy based network loading model that directly accepts policies as input. For this, we develop a novel LTM algorithm that is capable of representing spatial queues and shockwaves but accepts policies rather than paths as inputs. We conduct computational tests to study the convergence and computational performance of the proposed approach with SDTA solution using an iterative network loading scheme. We find that the proposed approach is more efficient than iterative network loading. \\ 

\noindent In the proposed SDTA solution algorithm, we keep the number of policies to be fixed in each iteration. The first policy is the optimal policy for the input travel time distribution and the remaining policies are generated using a novel algorithm inspired from link penalty heuristic of static shortest path problem. Although, there exist appropriation algorithms to generate suboptimal policies but their results can be arbitrarily worse. We propose an algorithm to generate suboptimal policies that involves perturbation in the elements of input travel time distribution. The advantage of our algorithm is that the quality of generated policies in comparison to optimal policy can be controlled. Also, we show the consistency of the developed algorithm by proving that the optimal policy is always allocated the largest traffic flow. As a consequence, we show that the flow proportion allocated to the optimal policy can be controlled by suitably varying a parameter in the developed suboptimal policy generation algorithm. Numerical tests conducted on sample networks also verify this proposition. We also come up with the sufficient conditions that allow largest flow allocation to the optimal policy if the developed algorithm performs perturbation in the travel time distribution in a different manner.  \\     

\noindent The solution existence of SDTA problem using fixed point theory is also discussed. Finally, we provide numerical results to show the benefits of using the proposed formulation in solving the SDTA problem for traffic networks. The insights from these results are discussed and illustrated. Also, there can be several extensions and related studies to the current study. For instance, the proposed SDTA model can be extended to multiple OD pairs with minor modifications. Improving the computational performance of SDTA algorithm using parallelization techniques also constitutes an important future study. The formulation of chronological network loading can also be extended to other spatial queue network loading models like cell transmission model. In addition, our optimal policy generation algorithm focuses on minimizing the expected travel time without focusing on the associated reliability. Reliability in the general context can be represented as the probability of satisfying a particular condition. This concept is useful when the objective is to satisfy some constraints but it is not possible to do because of the stochasticity present in the system and thus the objective becomes to maximize the probability of satisfying the constraints. In the context of travel time, reliability can be defined as the on-time arrival probability that a trip is successfully fulfilled within a desirable travel time budget \citep{sumalee2011stochastic,chen2014reliable,yang2017optimizing}. This extension will require maximizing the probability of a trip fulfilling a travel time budget along with minimizing the expected travel time in the optimal policy generation. Alternatively, travel time reliability can also be defined by simultaneously minimizing the variance and expected value of the travel time as suggested by \cite{gao2005optimal}.
%Commented acknowledgments for double blind process
%\section*{Acknowledgements}
%The authors are grateful to National Science Foundation for the award CMMI 1520338 to support the research presented in the paper. However, the authors are solely responsible for the findings presented in this study.

%\bibliographystyle{trb}
\bibliographystyle{apa}
\bibliography{LaTeX}  

\end{document}